\documentclass{article}

\usepackage{braket}
\usepackage{amsmath,amssymb,amsthm}
\usepackage{mathtools,bbm,nicefrac}
\usepackage{algorithm,algorithmic}
\usepackage{caption}
\allowdisplaybreaks[4]

\usepackage{color}

\DeclareMathOperator{\tr}{Tr}

\DeclareMathOperator{\spn}{span}
\DeclareMathOperator{\poly}{poly}
\DeclareMathOperator{\polylog}{polylog}

\DeclareMathOperator{\diam}{diam}

\newcommand{\q}{\mathbf{q}}
\newcommand{\Q}{\mathbf{Q}}
\newcommand{\SigmaOp}{\mathbf{\Sigma}}
\newcommand{\F}{\mathbf{F}}
\newcommand{\kernel}{\mathbf{k}}
\newcommand{\U}{\mathbf{U}}
\newcommand{\A}{\mathbf{A}}
\newcommand{\B}{\mathbf{B}}
\newcommand{\Lam}{\mathbf{\Lambda}}
\newcommand{\R}{\mathbf{R}}

\theoremstyle{plain}
\newtheorem{theorem}{Theorem}
\newtheorem{lemma}{Lemma}

\newtheorem{proposition}{Proposition}
\theoremstyle{definition}

\theoremstyle{remark}
\newtheorem{conjecture}{Conjecture}
\newtheorem{remark}[conjecture]{Remark}

\usepackage{microtype}
\usepackage{graphicx}
\usepackage{subfigure}
\usepackage{booktabs}

\PassOptionsToPackage{hyphens}{url}
\usepackage{hyperref}

\renewcommand\algorithmicrequire{\textbf{Input:}}
\renewcommand\algorithmicensure{\textbf{Output:}}

    \PassOptionsToPackage{numbers, sort&compress}{natbib}

    \usepackage[preprint]{neurips_2020}

\usepackage[utf8]{inputenc} 
\usepackage[T1]{fontenc}    

\title{Learning with Optimized Random Features: Exponential Speedup by Quantum Machine Learning without Sparsity and Low-Rank Assumptions}

\author{%
  Hayata Yamasaki\\
  The University of Tokyo,
  Austrian Academy of Sciences\\
  \texttt{hayata.yamasaki@gmail.com}
  \And%
  Sathyawageeswar Subramanian\\
  University of Cambridge,\\
  University of Warwick
  \AND%
  Sho Sonoda\\
  RIKEN AIP
  \And%
  Masato Koashi\\
  The University of Tokyo

}

\begin{document}

\maketitle

\begin{abstract}
Kernel methods augmented with random features give scalable algorithms for learning from big data.
But it has been computationally hard to sample random features according to a probability distribution that is optimized for the data, so as to minimize the required number of features for achieving the learning to  a desired accuracy.
Here, we develop a quantum algorithm for sampling from this optimized distribution over features, in runtime $O(D)$ that is linear in the dimension $D$ of the input data.
Our algorithm achieves an exponential speedup in $D$ compared to any known classical algorithm for this sampling task.
In contrast to existing quantum machine learning algorithms,
our algorithm circumvents sparsity and low-rank assumptions and thus has wide applicability.
We also show that the sampled features can be combined with regression by stochastic gradient descent to achieve the learning without canceling out our exponential speedup.
Our algorithm based on sampling optimized random features leads to an accelerated framework for machine learning that takes advantage of quantum computers.
\end{abstract}

\section{\label{sec:intro}Introduction}

Random features~\cite{R2,R3} provide a powerful technique for scaling up kernel methods~\cite{S5} applicable to various machine learning tasks, such as ridge regression~\cite{NIPS2017_6914}, kernel learning~\cite{NIPS2016_6180}, and principle component analysis~\cite{NIPS2018_7961}. Recently,~\citet{B1} has shown an optimized probability distribution of random features, and sampling features from this optimized distribution would drastically improve runtime of learning algorithms based on the random features.
However, this sampling task has been computationally ``hard in practice''~\cite{B1} due to inversion of a high-dimensional operator.
In contrast, the power of quantum computers to process data in quantum superposition attracts growing attention towards accelerating learning tasks, opening a new field: quantum machine learning (QML)~\cite{biamonte2017quantum,doi:10.1098/rspa.2017.0551,dunjko2018machine}.
In this work, we develop a framework of QML that accelerates a supervised learning task, by constructing an efficient quantum algorithm for sampling from this optimized distribution.

\textbf{Learning with random features:} Supervised learning deals with the problem of estimating an unknown function $y=f(x)$. We will consider $D$-dimensional input $x\in\mathbb{R}^D$ and real-valued output $y\in\mathbb{R}$.
Given $N$ input-output pairs of examples, we want to learn $f$ to a desired accuracy ${\epsilon}>0$.
To model $f$, kernel methods use reproducing kernel Hilbert space (RKHS) associated with a symmetric, positive semidefinite function $k(x^\prime,x)$, \emph{the kernel}~\cite{S5}.
Typical kernel methods that compute an $N\times N$ Gram matrix may not be scalable as $N$ gets large,
but random features~\cite{R2,R3}, along with other techniques via low-rank matrix approximation~\cite{10.5555/645529.657980,10.5555/3008751.3008847,10.5555/944790.944812}, enable scalable kernel-based algorithms.

Algorithms using random features are based on the fact that we can represent any translation-invariant kernel $k$ as expectation of a feature map $\varphi(v,x)=\mathrm{e}^{-2\pi\mathrm{i}v\cdot x}$ over a probability measure $d\tau(v)$ corresponding to the kernel.
Conventional algorithms using random features~\cite{R2,R3} sample $M$ $D$-dimensional parameters $v_0,\ldots,v_{M-1}\in\mathbb{R}^D$ from the distribution $d\tau(v)$ to determine $M$ features $\varphi(v_m,\cdot)$ in time $O(MD)$.
For a class of kernels such as Gaussian, this runtime may be reduced to $O(M\log D)$~\cite{L1,NIPS2016_6246}.
We learn the function $f$ using a linear combination of the $M$ features, i.e.,
\begin{equation}
  \label{eq:estimate}
  f(x)\approx \sum_{m=0}^{M-1}\alpha_m\varphi(v_m,x) \eqqcolon \hat{f}_{M,v_m,\alpha_m}(x).
\end{equation}
To achieve the learning to accuracy $O({\epsilon})$, we need to sample a sufficiently large number $M$ of features.
Once we fix $M$ features, we calculate coefficients $\alpha_m$ by linear (or ridge) regression to minimize an error between $f$ and $\hat{f}_{M,v_m,\alpha_m}$
using the $N$ given examples~\cite{carratino2018learning,R3,NIPS2017_6914}.
The sampling of features and the regression of coefficients can be performed simultaneously via doubly stochastic gradients~\cite{NIPS2014_5238}.

\textbf{Problem:} These conventional algorithms using random features sampled from the data-independent distribution $d\tau(v)$ require a \textit{large} number $M$ of features to learn the function $f$,
which \textit{slows down} the decision of all $M$ features and the regression over $M$ coefficients.
To improve this,
we aim to minimize $M$ required for the learning.
Rather than sampling from $d\tau(v)$,
we will sample features from a probability distribution that puts greater weight on \emph{important features optimized for the data} via a probability density function $q(v)$ for $d\tau(v)$.
To minimize $M$ achieving the accuracy $O({\epsilon})$,
\citet{B1} provides an optimized probability density function $q_{\epsilon}^\ast(v)$ for $d\tau(v)$ (see~\eqref{eq:q}, Sec.~\ref{sec:learning}).
This optimized $q_{\epsilon}^\ast(v)$ achieves minimal $M$ up to a logarithmic gap among all algorithms using random features for accuracy $\epsilon$~\cite{B1}.
It \textit{significantly improves} $M$ compared to sampling from $d\tau(v)$~\cite{B1,NIPS2017_6914,NIPS2018_7598}; e.g., to achieve learning with the Gaussian kernel from data given according to a sub-Gaussian distribution, compared to sampling from the data-independent distribution $d\tau(v)$ in Refs.~\cite{R2,R3}, the required number $M$ of features sampled from the optimized distribution $q_{\epsilon}^\ast(v)d\tau(v)$ can be \textit{exponentially small} in $\epsilon$~\cite{B1}.
We call features sampled from $q_{\epsilon}^\ast(v)d\tau(v)$ \textit{optimized random features}.

However, the sampling from $q_{\epsilon}^\ast(v)d\tau(v)$ has been ``hard in practice''~\cite{B1} for two reasons.
First, the definition~\eqref{eq:q} of $q_{\epsilon}^\ast(v)$ includes an \textit{infinite-dimensional operator} ${(\Sigma+{\epsilon}\mathbbm{1})}{}^{-1}$ on the space of functions $f:\mathbb{R}^D\to\mathbb{R}$ with $D$-dimensional input data, which is intractable to calculate by computer without approximation.
Second, even if we approximate ${\Sigma+{\epsilon}\mathbbm{1}}$ by an operator on a finite-dimensional space, the \textit{inverse operator} approximating ${(\Sigma+{\epsilon}\mathbbm{1})}{}^{-1}$ is still hard to calculate; in particular, for achieving a desired accuracy in the approximation, the required dimension of this finite-dimensional space can be exponentially large in $D$, i.e., $O(\exp(D))$~\cite{NIPS2018_7598,Shahrampour2019}, and \textit{no known algorithm} can calculate the inverse of the $O(\exp(D))$-dimensional operator in general \textit{within sub-exponential time} in $D$.

Note that Refs.~\cite{pmlr-v70-avron17a,Liu2019} propose probability density functions similar to $q_{\epsilon}^\ast(v)$,
from which the samples can be obtained in polynomial time~\cite{pmlr-v70-avron17a,pmlr-v97-li19k,Liu2019};
however,
in contrast to sampling from $q_{\epsilon}^\ast(v)d\tau(v)$,
sampling from the distributions in Refs.~\cite{pmlr-v70-avron17a,Liu2019} does not necessarily minimize the required number $M$ of features for the learning.
In particular, $q_{\epsilon}^\ast(v)d\tau(v)$ in Ref.~\cite{B1} and the distribution in Ref.~\cite{pmlr-v70-avron17a} are different in that the former is defined using an integral operator as shown in~\eqref{eq:q}, but the latter is defined using the Gram matrix; even if we discretize the integral operator, we do not obtain the Gram matrix. The distribution in Ref.~\cite{Liu2019} does not use the integral operator either. \citet{B1} proves optimality of $q_{\epsilon}^\ast(v)d\tau(v)$ in minimizing $M$ required for approximating function $f$, but this proof of the optimality is not applicable to the distributions in Refs.~\cite{pmlr-v70-avron17a,Liu2019}.
Similarly, whereas sampling from an importance-weighted distribution may also be used in column sampling for scaling-up kernel methods via low-rank matrix approximation, algorithms in the setting of the column sampling~\cite{pmlr-v30-Bach13,NIPS2015_5716,NIPS2018_7810} are not applicable to our setting of random features, as discussed in Ref.~\cite{B1}.
Quasi-Monte Carlo techniques~\cite{10.5555/2946645.3007073,10.5555/3172077.3172095} also improve $M$, but it is unknown whether they can achieve minimal $M$.

\textbf{Our contributions:}
As discussed above, the bottleneck in using random features sampled from the optimized distribution $q_{\epsilon}^\ast(v)d\tau(v)$ is each sampling step that works with inversion of $O(\exp(D))$-dimensional matrices for $D$-dimensional input data. To overcome this bottleneck and the difficulties in sampling from $q_{\epsilon}^\ast(v)d\tau(v)$,
we discover that we can use a \textit{quantum algorithm}, rather than conventional classical algorithms that run on existing computers.
Our contributions are as follows.

\begin{itemize}
  \item (Theorem~\ref{thm:sampling}) We construct a quantum algorithm for \textbf{sampling a feature from the \textit{data-optimized} distribution $q_{\epsilon}^\ast(v)d\tau(v)$ in as fast as linear runtime $O(D)$ in the data dimension} $D$.
    The best existing classical algorithm for sampling each single feature from \textit{data-optimized} $q_{\epsilon}^\ast(v)d\tau(v)$ requires exponential runtime $O(\exp(D))$~\cite{B1,NIPS2018_7598,Shahrampour2019}. In contrast, our quantum algorithm can sample the feature from $q_{\epsilon}^\ast(v)d\tau(v)$ in runtime $O(D)$, which is as fast as the conventional algorithms using random features~\cite{R2,R3}. We emphasize that the conventional algorithms perform an easier task, i.e., sampling from a \textit{data-independent} distribution $d\tau(v)$. Advantageously over the conventional algorithms sampling from $d\tau(v)$, we can use our algorithm sampling from $q_{\epsilon}^\ast(v)d\tau(v)$ to achieve learning with a significantly small number $M$ of features, which is proven to be \textit{minimal} up to a logarithmic gap~\cite{B1}.
    Remarkably, we achieve this without assuming sparsity or low rank of relevant operators.
  \item (Theorem~\ref{thm:complexity_all}) We show that we can combine $M$ features sampled by our algorithm with regression by stochastic gradient descent \textbf{to achieve supervised learning in time $O(MD)$, i.e., \textit{without canceling out our exponential speedup}}. This $M$ is minimal up to a logarithmic gap~\cite{B1} since we use optimized random features. Thus, by improving the computational bottleneck faced by classical algorithms for sampling optimized random features, we provide a promising framework of quantum machine learning that leverages our $O(D)$ sampling algorithm to achieve the optimal $M$ among all algorithms using random features.
\end{itemize}

\textbf{Comparison with previous works on quantum machine learning (QML):}
The novelty of our contributions is that we construct a QML algorithm that is \textit{exponentially faster} than any existing classical algorithm sampling from $q_{\epsilon}^\ast(v)d\tau(v)$~\cite{B1,NIPS2018_7598,Shahrampour2019}, yet is still \textit{free from sparsity and low-rank assumptions} on operators.
Despite major efforts to apply QML to kernel methods~\cite{Mengoni2019},
super-polynomial speedups like Shor's algorithm for prime factoring~\cite{10.1137/S0097539795293172} are rare in QML\@.
In fact, it has been challenging to find applications of quantum algorithms with super-polynomial speedups for practical problems~\cite{montanaro2016quantum}.
Typical QML algorithms such as Refs.~\cite{PhysRevLett.103.150502,PhysRevLett.109.050505,lloyd2016quantum,Z1} may achieve exponential speedups over classical algorithms only if matrices involved in the algorithms are sparse; in particular, $n\times n$ matrices can have only $\polylog(n)$ nonzero elements in each row and column. Another class of QML algorithms such as Refs.~\cite{lloyd2014quantum,PhysRevLett.113.130503,kerenidis_et_al:LIPIcs:2017:8154,PhysRevLett.120.050502} do not require sparsity but may attain large speedups only if the $n\times n$ matrices have low rank $\polylog(n)$. This class of quantum algorithms are polynomially faster than recent ``quantum-inspired'' classical algorithms such as Refs.~\cite{10.1145/3313276.3316310,arXiv:1910.05699,arXiv:1910.06151}, which also assume low rank. Quantum singular value transformation (QSVT)~\cite{G1} has recently emerged as a fundamental subroutine to implement these quantum algorithms in a unified way.
However, power and applicability of these QML algorithms are restricted by the extreme assumptions on sparsity and low rank~\cite{aaronson2015read}.

Our key technical contribution is to develop an approach for circumventing the sparsity and low-rank assumptions in our QML algorithm, broadening the applicability of QML\@.
We achieve this by combining the QSVT with another fundamental subroutine, quantum Fourier transform (QFT)~\cite{C5,H2}.
QFT and QSVT are commonly used in quantum computation~\cite{N4,arXiv:1907.09415};
however, it is nontrivial to use these subroutines for developing a QML algorithm that exponentially outperforms existing classical algorithms under widely applicable assumptions.
To achieve the speedup, our technique decomposes the $O(\exp(D))$-dimensional \textit{non-sparse and full-rank} operator representing ${\Sigma+{\epsilon}\mathbbm{1}}$ in the definition~\eqref{eq:q} of $q_{\epsilon}^\ast(v)$ into diagonal (i.e., sparse) operators using Fourier transform.
QSVT and QFT may make our algorithm hard to simulate by classical computation, and hard to perform even on near-term quantum devices~\cite{havlivcek2019supervised,arute2019quantum} that cannot implement universal quantum computation due to noise.
For this reason, this paper does not include numerical simulation, and we analytically prove the runtime of our algorithm.
In contrast to heuristic QML algorithms for noisy quantum devices such as Ref.~\cite{havlivcek2019supervised} where no proof bounds its runtime, our QML algorithm aims at applications on large scales; to achieve this aim, our proof shows the exponential advantage of our quantum algorithm over the existing classical algorithms in terms of the runtime.
The wide applicability of our QML algorithm makes it a promising candidate for ``killer applications'' of universal quantum computers in the long run; after all, large-scale machine learning will be eventually needed in practice.

Also remarkably, since we exploit quantum computation for the sampling problem of $q_{\epsilon}^\ast(v)d\tau(v)$, our QML algorithm avoids overhead of repeating preparation of quantum states many times for estimating expectation values from the states.
The classical algorithm~\cite{B1,NIPS2018_7598,Shahrampour2019} calculates $q_{\epsilon}^\ast(v)d\tau(v)$ by matrix inversion and then performs sampling;
in contrast, our quantum algorithm never estimates classical description of $q_{\epsilon}^\ast(v)d\tau(v)$, which is represented by amplitude of a quantum state, since the overhead of such estimation would cancel out the speedup~\cite{aaronson2015read}.
Instead, our exponential quantum speedup is achieved by performing a quantum measurement of this state to sample a feature efficiently per single preparation and measurement of the state.
Our algorithm combined with stochastic gradient descent provides classical description of an estimate of $f$ to be learned, rather than a quantum state.
In this way, our results discover an application of sampling problems to designing fast QML algorithms.

\section{\label{sec:assumptions}Setting of learning with optimized random features}

\subsection{Notation on quantum computation}

In Supplementary Material, we summarize basic notions and notations of quantum computation to describe our quantum algorithms,
referring to Refs.~\cite{N4,arXiv:1907.09415} for more detail.
An $m$-qubit quantum register is represented as a $2^m$-dimensional Hilbert space $\mathcal{H}={(\mathbb{C}^2)}^{\otimes m}$. Following the conventional bra-ket notation, we represent a quantum state on the quantum register as a ket (i.e., a vector) $\Ket{\psi}\in\mathcal{H}$.

\subsection{\label{sec:learning}Supervised learning with optimized random features}

We introduce the supervised learning setting that we focus on in this paper,
and we will formulate an approximate version of it in Sec.~\ref{sec:coarse_grained}.
Suppose that $N$ input-output pairs of examples are given by
$(x_0,y_0),\ldots,(x_{N-1},y_{N-1})\in\mathcal{X}\times\mathcal{Y}$,
where $y_n=f(x_n)$, $f:\mathcal{X}\to\mathcal{Y}$ is an unknown function to be learned, $\mathcal{X}=\mathbb{R}^D$ is the domain for $D$-dimensional input data, $\mathcal{Y}=\mathbb{R}$ is the range for output data.
Each $x_n$ is an observation of an independently and identically distributed (IID) random variable on $\mathcal{X}$ equipped with a probability measure $d\rho(x)=q^{(\rho)}(x)dx$, on which we will pose a mild assumption later in Sec.~\ref{sec:coarse_grained}.
We choose a translation-invariant kernel, and such a kernel can be represented as~\cite{R2}
\begin{equation}
  \label{eq:kernel}
k(x^\prime,x)= \int d\tau(v)\overline{\varphi(v,x^\prime)}\varphi(v,x),\,\big(\text{normalized by $k(x,x)=k(0,0)=\int_\mathcal{V}d\tau\left(v\right)=1$}\big),
\end{equation}
where $\overline{\,\cdot\,}$ is complex conjugation,
$\varphi:\mathcal{V}\times\mathcal{X}\to\mathbb{C}$ is a feature map $\varphi(v,x)=\mathrm{e}^{-2\pi\mathrm{i}v\cdot x}$, $\mathcal{V}=\mathbb{R}^D$ is a parameter space equipped with a probability measure $d\tau(v)=q^{(\tau)}(v)dv$,
and $d\tau(v)$ is given by the Fourier transform of $k$.
To specify a model of $f$,
we use the RKHS $\mathcal{F}$ associated with the kernel $k$~\cite{S5}.
Following Ref.~\cite{B1},
we assume that the norm $\|f\|_\mathcal{F}$ of $f$ in the RKHS is bounded, in particular, $\|f\|_\mathcal{F}\leqq 1$.
We aim to learn an estimate $\hat{f}$ of $f$ from the $N$ given examples of data, so that the generalization error $\int d\rho(x)|\hat{f}(x)-f(x)|^2$ can be bounded to a desired learning accuracy ${\epsilon}>0$.

To achieve the learning to accuracy $O({\epsilon})$ with the minimal number $M$ of features,
instead of sampling from $d\tau$,
\citet{B1} proposes to sample features from an optimized probability density $q_{\epsilon}^\ast$ for $d\tau$
\begin{align}
  \label{eq:q}
  q_{\epsilon}^\ast\left(v\right)\propto{\braket{\varphi\left(v,\cdot\right)|{\left(\Sigma+{\epsilon}\mathbbm{1}\right)}^{-1}\varphi\left(v,\cdot\right)}_{L_2(d\rho)}},\,\big(\text{normalized by $\int_\mathcal{V}q_{\epsilon}^\ast\left(v\right)d\tau(v)=1$}\big),
\end{align}
where $\Braket{f|g}_{L_2(d\rho)}\coloneqq\int_{\mathcal{X}}d\rho(x)\overline{f(x)}g(x)$, $\mathbbm{1}$ is the identity operator,
and $\Sigma:L_2\left(d\rho\right)\to L_2\left(d\rho\right)$ is the integral operator $\left(\Sigma f\right)\left(x^\prime\right)\coloneqq\int_\mathcal{X}d\rho\left(x\right)\,k\left(x^\prime,x\right)f\left(x\right)$~\cite{C2}.
The function $q_{\epsilon}^\ast(v)$ is called a leverage score.
Then,
it suffices to sample $M$ features
from $q_{\epsilon}^\ast(v)d\tau(v)$ with $M$ bounded by~\cite{B1}
\begin{equation}
  \label{eq:M_bound}
  M=O(d\left({\epsilon}\right)\log\left(\nicefrac{d\left({\epsilon}\right)}{\delta}\right)),
\end{equation}
so as to achieve the learning to accuracy $O(\epsilon)$ with high probability greater than $1-\delta$ for any $f$ satisfying $\|f\|_\mathcal{F}\leqq 1$,
in formula,
$\min_{\alpha_m}\{\int d\rho(x)|\hat{f}_{M,v_m,\alpha_m}(x)-f(x)|^2\}\leqq 4\epsilon$,
where
$\hat{f}_{M,v_m,\alpha_m}$ is the estimate~\eqref{eq:estimate} of $f$,
and $d\left({\epsilon}\right)\coloneqq\tr\Sigma{\left(\Sigma+{\epsilon}\mathbbm{1}\right)}^{-1}$ is \textit{the degree of freedom} representing effective dimension of data.
In this paper,
features sampled from $q_{\epsilon}^\ast(v)d\tau(v)$ up to approximation are called \textit{optimized random features}, which achieve the learning with minimal $M$
up to a logarithmic gap~\cite{B1}.
In kernel methods,
kernel $k$ should be chosen suitably to learn from the data given according to the distribution $d\rho$; otherwise, it is impossible for the kernel methods to achieve the learning with reasonable runtime and accuracy.
In the case of random features, $f$ must have a polynomial-size description in terms of the features, i.e., $M=O(\poly(D,\nicefrac{1}{\epsilon}))$. To guarantee this, the degree of freedom $d\left({\epsilon}\right)$ must satisfy
\begin{equation}
  \label{eq:effective_dim_bound}
  d\left({\epsilon}\right)=O(\poly(D,\nicefrac{1}{\epsilon})),
\end{equation}
where $d\left({\epsilon}\right)$ depends on $\Sigma$ and hence on both $d\rho$ and $k$ that is to be chosen suitably to satisfy~\eqref{eq:effective_dim_bound}.

\subsection{\label{sec:coarse_grained}Discretized representation of real number}

\begin{table}[t]
  \vskip -4mm
  \begin{minipage}[t]{.40\textwidth}
    \caption{\label{table:rescaling}Rescaling data by $r>1$.}
    \begin{center}
      \begin{small}
        \renewcommand{\arraystretch}{1.1}
        \begin{tabular}{llccr}
          \toprule
          Original & Rescaled by $r>1$ \\
          \midrule
          $G$ of interval $[0,G]$ & $G_r=rG$\\
          Kernel $\tilde{k}(x^\prime,x)$ & $\tilde{k}_r(rx^\prime,rx)\coloneqq \tilde{k}(x^\prime,x)$\\
          Input $x$ & $rx$ \\
          Output $y=f(x)$ & $y=f_r(rx)\coloneqq f(x)$ \\
          $q^{(\rho)}(x)dx=d\rho(x)$ & $q^{(\rho)}_r(rx)\coloneqq q^{(\rho)}(x)/r$ \\
          $f(x)$'s LC $L^{(f)}$ & $L_r^{(f_r)}=\nicefrac{L^{(f)}}{r}$\\
          $q^{(\rho)}(x)$'s LC $L^{(q^{(\rho)})}$ & $L_r^{(q_r^{(\rho)})}=\nicefrac{L^{(q^{(\rho)})}}{r^2}$\\
          \bottomrule
        \end{tabular}
      \end{small}
    \end{center}
  \end{minipage}
  \hfill
  \begin{minipage}[t]{.50\textwidth}
    \caption{\label{table:coarse_graining}Discretized representation ($\tilde{x}^\prime,\tilde{x}\in\tilde{\mathcal{X}}$).}
    \begin{center}
      \begin{small}
        \renewcommand{\arraystretch}{1.2}
        \begin{tabular}{llccr}
          \toprule
          Function / operator on $\mathcal{X}$ & Vector / operator on $\mathcal{H}^X$ \\
          \midrule
          $f:\mathcal{X}\to\mathbb{C}$ & $\Ket{f}\coloneqq\textstyle\sum_{\tilde{x}}f(\tilde{x})\Ket{\tilde{x}}$ \\
          $\varphi(v,\cdot):\mathcal{X}\to\mathbb{C}$ & $\Ket{\varphi(v,\cdot)}\coloneqq\sum_{\tilde{x}}\varphi(v,\tilde{x})\Ket{\tilde{x}}$ \\
          $\tilde{k}:\mathcal{X}\times\mathcal{X}\to\mathbb{R}$ & $\kernel\coloneqq\sum_{\tilde{x}^\prime,\tilde{x}}\tilde{k}(\tilde{x}^\prime,\tilde{x})\Ket{\tilde{x}^\prime}\Bra{\tilde{x}}$ \\
          $q^{(\rho)}:\mathcal{X}\to\mathbb{R}$ & $\q^{(\rho)}\coloneqq\sum_{\tilde{x}}q^{(\rho)}(\tilde{x})\Ket{\tilde{x}}\Bra{\tilde{x}}$ \\
          $\Sigma$ acting on $f:\mathcal{X}\to\mathbb{C}$ & $\SigmaOp\coloneqq \kernel \q^{(\rho)}$ \\
          $\Sigma f:\mathcal{X}\to\mathbb{C}$ & $\SigmaOp\Ket{f}$\\
          $\hat{q}^{(\rho)}:\tilde{\mathcal{X}}\to\mathbb{R}$ (Sec.~\ref{sec:data}) & $\hat{\q}^{(\rho)}\coloneqq\sum_{\tilde{x}}\hat{q}^{(\rho)}(\tilde{x})\Ket{\tilde{x}}\Bra{\tilde{x}}$ \\
          \bottomrule
        \end{tabular}
      \end{small}
    \end{center}
  \end{minipage}
  \vskip -6mm
\end{table}

To clarify our setting of digital quantum computation, we explain discretized representation of real number used in our quantum algorithm.
We assume that the input data domain is bounded;
in particular, the data distribution $d\rho\left(x\right)$ is nonzero only on a bounded domain ${[0,x_{\max}]}^D$ ($x_{\max}>0$).
If the kernel $k(x^\prime,x)$, such as Gaussian, decays to $0$ sufficiently fast as $x^\prime$ and $x$ deviate from $0$, then we can approximate $k(x^\prime,x)$  using a periodic function $\tilde{k}$ with a sufficiently large period $G\gg x_{\max}$
\begin{equation}
\label{eq:kernel_approximation}
k(x^\prime,x)\approx\sum_{n\in\mathbb{Z}^D}k(x^\prime,x+Gn)\eqqcolon \tilde{k}(x^\prime,x),\quad\forall x^\prime,x\in{[0,x_{\max}]}^D.
\end{equation}
We will use $\tilde{k}$ as a kernel in place of $k$.
In computation, it is usual to represent a real number using a finite number of bits; e.g., fixed-point number representation with small precision $\Delta>0$ uses a finite set $\{0,\Delta,2\Delta,\ldots,{G}-\Delta\}$ to represent a real interval $[0,G]$.
Equivalently, to simplify the presentation, we use the fixed-point number representation rescaled by a parameter $r=\nicefrac{1}{\Delta}$ as shown in Table~\ref{table:rescaling}, so that we can use a set of integers $\mathcal{I}=\{0,1,\ldots,G_r-1\}$ to discretize the interval.
We represent the data domain $\mathcal{X}=\mathbb{R}^D$ as
$\tilde{\mathcal{X}}=\mathcal{I}^D$.
Discretization of the data range $\mathcal{Y}$ is unnecessary in this paper.
For any real-valued point $x\in\mathcal{X}$,
we write its closest grid point as $\tilde{x}\in\tilde{\mathcal{X}}$,
and let $\Delta_x\subset\mathbb{R}^D$ denote a $D$-dimensional unit hypercube whose center is the closest grid point $\tilde{x}$ to $x$.

To justify this discretization, we assume that functions in the learning, such as the function $f$ to be learned and the probability density $q^{(\rho)}(x)$ of input data, are $L$-Lipschitz continuous for some Lipschitz constant (LC) $L$.\footnote{%
For any $x,x^\prime\in\mathcal{X}$, a function $q:\mathcal{X}\to\mathbb{C}$ is $L$-Lipschitz continuous if $|q\left(x\right)-q\left(x^\prime\right)|\leqq L\|x-x^\prime\|_2$.
}
Then, errors caused by the discretization, i.e., $|f(x)-f(\tilde{x})|$ and $|q^{(\rho)}(x)-q^{(\rho)}(\tilde{x})|$, are negligible in the limit of small (but still nonzero) $L$, in particular, $L\sqrt{D}\to 0$.
As the data dimension $D$ gets large,
to reduce $L\sqrt{D}$ to a fixed error threshold,
we rescale the data to a larger domain (see Table~\ref{table:rescaling});
in particular, we rescale $G$ representing the interval $[0,G]$ to ${G}_r=rG$ with $r=\Omega(L\sqrt{D})$.
The rescaling in Table~\ref{table:rescaling} keeps the accuracy and the model in the learning \textit{invariant}.

We focus on asymptotic runtime analysis of our algorithm as ${G}_r$ gets larger, i.e., $G_r\to\infty$, which reduces the errors in the discretization. We henceforth omit the subscript $r$ and write ${G}_r$ as ${G}$ for brevity.
An error analysis of discretization for finite ${G}$ is out of the scope of this paper; for such an analysis, we refer to established procedures in signal processing~\cite{proakis2001digital}.

As we can represent $\tilde{\mathcal{X}}$ using $D\lceil\log_2 {G}\rceil$ bits,
where $\lceil x \rceil$ is the least integer greater than or equal to $x$,
we similarly represent $\tilde{\mathcal{X}}$ using a quantum register $\mathcal{H}^{X}\coloneqq\spn\{\Ket{\tilde{x}}:\tilde{x}\in\tilde{\mathcal{X}}\}$ of $D\lceil\log_2 {G}\rceil$ qubits.
This quantum register is composed of $D$ sub-registers, i.e.,
$\mathcal{H}^{X}={\left(\mathcal{H}_{\mathcal{I}}\right)}^{\otimes D}$,
where each $\lceil\log_2 {G}\rceil$-qubit sub-register $\mathcal{H}_{\mathcal{I}}={\left(\mathbb{C}^{2}\right)}{}^{\otimes \left\lceil\log_2 {G}\right\rceil}$ corresponds to $\mathcal{I}$.
To represent $\tilde{x}={(\tilde{x}^{\left(1\right)},\ldots,\tilde{x}^{\left(D\right)})}{}^\mathrm{T}\in\tilde{\mathcal{X}}$,
we use a quantum state
$\Ket{\tilde{x}}^X=\bigotimes_{d=1}^D\Ket{\tilde{x}^{\left(d\right)}}\in\mathcal{H}^X$,
where $\Ket{\tilde{x}^{\left(d\right)}}\in\mathcal{H}_{\mathcal{I}}$.

We represent a function on the continuous space $\mathcal{X}$ as a vector on finite-dimensional $\mathcal{H}^X$, and an operator acting on functions on $\mathcal{X}$ as a matrix on $\mathcal{H}^X$, as shown in Table~\ref{table:coarse_graining}.
Under our assumption that the rescaling makes the Lipschitz constants sufficiently small,
we can make an approximation
\begin{equation}
  \braket{f|\q^{(\rho)}|g}\approx\int_\mathcal{X}d\rho(x)\,\overline{f\left(x\right)}g\left(x\right),\; q^{(\rho)}(x)dx=d\rho(x).
\end{equation}
With this discretization, we can represent the optimized probability density function $q_{\epsilon}^\ast$ in~\eqref{eq:q}
as
\begin{equation}
  \label{eq:tilde_q_lambda_ast}
  \tilde{q}_{\epsilon}^\ast\left(v\right)\propto\Bra{\varphi\left(v,\cdot\right)}\q^{\left(\rho\right)}{\left(\SigmaOp+{\epsilon}\mathbbm{1}\right)}^{-1}\Ket{\varphi\left(v,\cdot\right)},\,\big(\text{normalized by $\int_\mathcal{V}\tilde{q}_{\epsilon}^\ast\left(v\right)d\tau(v)=1$}\big).
\end{equation}

\subsection{\label{sec:data}Data in discretized representation}

To represent real-valued input data $x_n\in\mathcal{X}$ that is IID sampled according to the probability measure $d\rho(x)$, we use discretization.
We represent $x_n$ using its closest (i.e., rounded) grid point $\tilde{x}_n\in\tilde{\mathcal{X}}$, IID sampled with probability
$\int_{\Delta_{x_n}}d\rho\left(x\right)$,
where $\Delta_{x_n}$ is the $D$-dimensional unit hypercube centered at $x_n$.
This rounding may cause some error in learning but does not significantly ruin the performance of our QML algorithm; after all, any implementation of kernel methods by computer with bits requires rounding, and in our setting, a cluster of points that would be represented as the same grid point after the rounding are resolved by rescaling, which is equivalent to increasing precision of rounding without rescaling. Then under standard assumptions in signal processing~\cite{proakis2001digital} where such implementation works well, it should be straightforward to show our algorithm also works well.
In the following, the $N$ given examples are $(\tilde{x}_0,y_0),\ldots,(\tilde{x}_{N-1},y_{N-1})\in\tilde{\mathcal{X}}\times\mathcal{Y}$, where $y_n=f(\tilde{x}_{n})$.

The true probability distribution $d\rho$ of the input data is unknown in our setting,
and our algorithm uses the $N$ given examples of data to approximate $d\rho(x)=q^{(\rho)}(x)dx$ up to a statistical error.
For any $\tilde{x}\in\tilde{\mathcal{X}}$,
we approximate the distribution $d\rho$ near $\tilde{x}$ by an empirical distribution counting the data: $\hat{q}^{\left(\rho\right)}\left(\tilde{x}\right)\coloneqq\nicefrac{n\left(\tilde{x}\right)}{N}$, where $n(\tilde{x})$ denotes the number of given examples of input data that are included in the $D$-dimensional unit hypercube $\Delta_{\tilde{x}}$.
We also represent $\hat{q}^{\left(\rho\right)}$ as an operator $\hat{\q}^{\left(\rho\right)}$ shown in Table~\ref{table:coarse_graining}.
In the same way as $\SigmaOp=\kernel \q^{(\rho)}$ in Table~\ref{table:coarse_graining},
an empirical integral operator is given by
\begin{equation}
\hat{\SigmaOp}\coloneqq \kernel\hat{\q}^{\left(\rho\right)}.
\end{equation}

We aim to analyze the asymptotic runtime of our algorithm when the number $N$ of examples becomes large, as with analyzing the cases of large $G$ in the rescaling.
In the limit of $N\to\infty$,
statistical errors in the empirical distribution caused by the finiteness of $N$ vanish.
Analysis of statistical errors for finite $N$ is out of the scope of this paper; for such an analysis, see Ref.~\cite{B1}.

\section{Learning with optimized random features}

We now describe our efficient quantum algorithm for sampling an optimized random feature in the setting of Sec.~\ref{sec:assumptions}.
As we show in Sec.~\ref{sec:perfect_reconstruction},
the novelty of our algorithm is to achieve this without assuming sparsity and low rank by means of the \textit{perfect reconstruction} of the kernel, which decomposes the kernel by Fourier transform into a finite sum of the feature map $\varphi$ weighted over a finite set of features.
In Sec.~\ref{sec:assumption_kernel},
we clarify assumptions on our quantum algorithm, bound its runtime (Theorem~\ref{thm:sampling}),
and also show that we can achieve the learning as a whole without canceling out our quantum speedup by combining our quantum algorithm with stochastic gradient descent (Theorem~\ref{thm:complexity_all}).

Compared to existing works~\cite{pmlr-v70-avron17a,Liu2019,pmlr-v97-li19k} on sampling random features from weighted probability distributions for acceleration, the significance of our results is that our algorithm in the limit of good approximation (as $N,G\to\infty$) is provably optimal in terms of a gap from a lower bound of the required number of random features for achieving learning~\cite{B1}, and yet its runtime is as fast as linear in $D$ and poly-logarithmic in $G$ (and $N$).\footnote{%
  The runtime shown in Theorems~\ref{thm:sampling} and~\ref{thm:complexity_all} is constant time in $N$ except that classical and quantum oracles that abstract devices for accessing data may have runtime $O(1)$ or $O(\polylog(N))$, as discussed in Sec.~\ref{sec:assumption_kernel}.}
Our algorithm is constructed so as to converge to sampling from the optimized distribution~\eqref{eq:q} in Ref.~\cite{B1} as $N,G\to\infty$ whereas the algorithms in Refs.~\cite{pmlr-v70-avron17a,Liu2019,pmlr-v97-li19k} do not converge to sampling from~\eqref{eq:q} in any limit. Although the algorithms in Refs.~\cite{pmlr-v70-avron17a,Liu2019,pmlr-v97-li19k} can achieve learning, the optimality of Refs.~\cite{pmlr-v70-avron17a,Liu2019,pmlr-v97-li19k} is unknown in general; in contrast, Ref.~\cite{B1} proves the optimality up to a logarithmic gap, and our algorithm based on Ref.~\cite{B1} achieves this optimality in the limit of $N,G\to\infty$.

\subsection{\label{sec:perfect_reconstruction}Main idea of quantum algorithm for sampling an optimized random feature}

The crucial technique in our quantum algorithm is to use the perfect reconstruction of the kernel (See Proposition~1 in Supplementary Material).
In the same way as representing the kernel $k$ as the expectation~\eqref{eq:kernel} of $\varphi(v,x)=\mathrm{e}^{-2\pi\mathrm{i}v\cdot x}$ over the probability distribution $d\tau=q^{(\tau)}(v)dv$,
we represent our kernel $\tilde{k}$ using Shannon's sampling theorem~\cite{shannon1949communication} in signal processing as
\begin{equation}
  \tilde{k}(x^\prime,x)=\sum_{\tilde{v}\in \mathbb{Z}^D}(\nicefrac{q^{(\tau)}(\nicefrac{\tilde{v}}{G})}{G^D})\overline{\varphi(\nicefrac{\tilde{v}}{G},x^\prime)}\varphi(\nicefrac{\tilde{v}}{G},x).
\end{equation}
Moreover, we show that to represent $\tilde{k}$ exactly on our \textit{discrete} data domain $\tilde{\mathcal{X}}$, it suffices to use a \textit{finite} set $\mathcal{V}_{{G}}$ of features
and a distribution function $Q^{(\tau)}$ over the finite set $\mathcal{V}_G$
\begin{equation}
\label{eq:finite_feature_set}
Q^{(\tau)}\left(v_{{G}}\right)\coloneqq\sum_{\tilde{v}^\prime\in\mathbb{Z}^D}q^{(\tau)}\left(v_{{G}}+\tilde{v}^\prime\right),\quad
v_G\in\mathcal{V}_{{G}}\coloneqq{\{0,\nicefrac{1}{{G}},\ldots,1-\nicefrac{1}{{G}}\}}^D,
\end{equation}
where we give examples of $Q^{(\tau)}$ in Table~\ref{table:kernel}.
In particular, for all $\tilde{x}^\prime,\tilde{x}\in\tilde{\mathcal{X}}$,
we show the following perfect reconstruction of our kernel $\tilde{k}(\tilde{x}^\prime,\tilde{x})$ from the function $Q^{(\tau)}$ using $D$-dimensional discrete Fourier transform $\F_D$ and its inverse $\F_D^\dag$\footnote{%
  With $\F$ denoting a unitary operator of (one-dimensional) discrete Fourier transform,
  we define $\F_D\coloneqq \F^{\otimes D}$.
}
\begin{equation}
  \label{eq:perfect_reconstruction}
  \tilde{k}\left(\tilde{x}^\prime,\tilde{x}\right)=\sum_{v_G\in\mathcal{V}_G}(\nicefrac{Q^{(\tau)}(v_G)}{G^D})\overline{\varphi(v_G,\tilde{x}^\prime)}\varphi(v_G,\tilde{x})
  \left(=\Bra{\tilde{x}^\prime}\F_D^\dag \Q^{(\tau)} \F_D\Ket{\tilde{x}}=\Bra{\tilde{x}^\prime}\F_D \Q^{(\tau)} \F_D^\dag\Ket{\tilde{x}}\right),
\end{equation}
where $\Q^{(\tau)}\coloneqq\sum_{\tilde{x}\in\tilde{\mathcal{X}}}Q^{(\tau)}\left(v_G\right)\Ket{\tilde{x}}\Bra{\tilde{x}}$ with $v_G=\nicefrac{\tilde{x}}{G}$ is a diagonal operator representing $Q^{(\tau)}$.

\begin{table}[tb]
  \vskip -4mm
  \caption{\label{table:kernel}Distribution $Q^{(\tau)}(v_G)$ for the Gaussian kernel (top) and the Laplacian kernel (bottom), where $v_{{G}}={(v_{{G}}^{(1)},\ldots,v_{{G}}^{(D)})}^\mathrm{T}$, and $\vartheta\left(u;q\right)\coloneqq 1+2\sum_{n=1}^{\infty}q^{n^2}\cos\left(2nu\right)$ is the theta function.}
  \begin{center}
    \begin{small}
      \renewcommand{\arraystretch}{1.1}
      \begin{tabular}{lllcr}
        \toprule
        $k(x^\prime,x)$ & $Q^{(\tau)}(v_G)$ \\
        \midrule
        Gaussian kernel: $\exp(-\gamma\left\|x^\prime-x\right\|_2^2)$ & $\prod_{d=1}^D\vartheta(\pi v_{{G}}^{(d)};\exp(-\gamma))$\\
        Laplacian kernel: $\exp(-\gamma\left\|x^\prime-x\right\|_1)$ & $\prod_{d=1}^D\nicefrac{\sinh(\gamma)}{(\cosh(\gamma)-\cos(2\pi v_{G}^{(d)}))}$\\
        \bottomrule
      \end{tabular}
    \end{small}
  \end{center}
  \vskip -6mm
\end{table}

Thus, similarly to conventional random features using Fourier transform~\cite{R2},
if we sampled a sufficiently large number $M$ of features in $\mathcal{V}_{{G}}$ from the probability mass function $P^{(\tau)}(v_G)\coloneqq\nicefrac{Q^{(\tau)}(v_G)}{\big(\sum_{v_G^\prime\in\mathcal{V}_G}Q^{(\tau)}(v_G^\prime)\big)}$ corresponding to $d\tau$,
then we could combine the $M$ features with the discrete Fourier transform $\F_D$ to achieve the learning with the kernel $\tilde{k}(\tilde{x}^\prime,\tilde{x})$.
However, $P^{(\tau)}(v_G)$ is not optimized for the data,
and our quantum algorithm aims to minimize $M$ by sampling an optimized random feature.
To achieve this,
in place of the optimized density $\tilde{q}_\epsilon^\ast$ defined as~\eqref{eq:tilde_q_lambda_ast} for $d\tau$ on the set $\mathcal{V}$ of real-valued features,
we define an optimized probability density function $Q_{\epsilon}^\ast(v_{{G}})$ for weighting the probability distribution $P^{(\tau)}(v_{{G}})$ on the finite set $\mathcal{V}_{{G}}$ of our features as
\begin{equation}
  \label{eq:hat_p_lambda_ast}
  Q_{\epsilon}^\ast(v_{{G}})\propto\braket{\varphi(v_{{G}},\cdot) | \hat{\q}^{(\rho)}{({\hat{\SigmaOp}}+{\epsilon}\mathbbm{1})}^{-1} | \varphi(v_{{G}},\cdot)},\,\big(\text{normalized by $\sum_{v_{{G}}\in\mathcal{V}_{G}}Q_{\epsilon}^\ast(v_{{G}})P^{(\tau)}(v_G)=1$}\big).
\end{equation}

To sample from optimized $Q_{\epsilon}^\ast(v_{{G}})P^{(\tau)}(v_{{G}})$,
we show that we can use a quantum state on two registers $\mathcal{H}^X\otimes\mathcal{H}^{X^\prime}$ of the same number of qubits (See Proposition~2 in Supplementary Material)
\begin{equation}
\label{eq:psi}
\Ket{\Psi}^{XX^\prime}\!\!\propto\!\!\sum_{\tilde{x}\in\tilde{\mathcal{X}}}\hat{\SigmaOp}_{\epsilon}^{-\frac{1}{2}}\Ket{\tilde{x}}^X\otimes \sqrt{(\nicefrac{1}{Q_{\max}^{(\tau)}})\Q^{(\tau)}}\F_D^\dag\sqrt{\hat{q}^{\left(\rho\right)}(\tilde{x})}\Ket{\tilde{x}}^{X^\prime},
\end{equation}
where
${Q_{\max}^{(\tau)}}\coloneqq\max\{Q^{(\tau)}\left(v_{{G}}\right):v_{{G}}\in\mathcal{V}_{{G}}\}$
is the maximum of $Q^{(\tau)}(v_G)$,
$\hat{\SigmaOp}_{\epsilon}$ is a positive semidefinite operator $\hat{\SigmaOp}_{\epsilon}\coloneqq(\nicefrac{1}{Q_{\max}^{(\tau)}})\sqrt{\hat{\q}^{\left(\rho\right)}}\kernel\sqrt{\hat{\q}^{\left(\rho\right)}}+(\nicefrac{{\epsilon}}{Q_{\max}^{(\tau)}})\mathbbm{1}$,
and $f(\A)$ for an operator $\A$ denotes an operator given by applying $f$ to the singular values of $\A$ while keeping the singular vectors, e.g., $\sqrt{\hat{\q}^{(\rho)}}=\sum_{\tilde{x}\in\tilde{\mathcal{X}}}\sqrt{\hat{q}^{(\rho)}\left(\tilde{x}\right)}\Ket{\tilde{x}}\Bra{\tilde{x}}$.
We show that if we perform a quantum measurement of the register $\mathcal{H}^{X^\prime}$ for the state $\Ket{\Psi}^{XX^\prime}$ in the computational basis $\{\Ket{\tilde{x}}^{X^\prime}\}$, we obtain a measurement outcome $\tilde{x}$ with probability $Q_{\epsilon}^\ast(\nicefrac{\tilde{x}}{{G}})P^{(\tau)}(\nicefrac{\tilde{x}}{{G}})$.
Our quantum algorithm prepares $\Ket{\Psi}^{XX^\prime}$ efficiently, followed by the measurement to achieve the sampling from $Q_{\epsilon}^\ast(v_{{G}})P^{(\tau)}(v_{{G}})$, where $v_{G}=\nicefrac{\tilde{x}}{G}$.

The difficulty in preparing the state $\Ket{\Psi}^{XX^\prime}$ arises from the fact that $\Ket{\Psi}$ in~\eqref{eq:psi} includes a ${G}^D$-dimensional operator $\hat{\SigmaOp}_{\epsilon}^{-\frac{1}{2}}$, i.e.\ on an \textit{exponentially large} space in $D$, and $\hat{\SigmaOp}_{\epsilon}$ may \textit{not be sparse or of low rank}.
One way to use a linear operator, such as $\hat{\SigmaOp}_{\epsilon}$ and $\hat{\SigmaOp}_{\epsilon}^{-\frac{1}{2}}$, in quantum computation is to use the technique of block encoding~\cite{G1}.
In conventional ways, we can efficiently implement block encodings of \textit{sparse or low-rank} operators~\cite{G1}, such as the diagonal operator $\sqrt{(\nicefrac{1}{Q_{\max}^{(\tau)}})\Q^{(\tau)}}$ in~\eqref{eq:psi}.
If we had an efficient implementation of a block encoding of $\hat{\SigmaOp}_{\epsilon}$, quantum singular value transformation (QSVT)~\cite{G1} would give an efficient way to implement a block encoding of $\hat{\SigmaOp}_{\epsilon}^{-\frac{1}{2}}$ to prepare $\Ket{\Psi}$.
However, it has not been straightforward to discover such an efficient implementation for $\hat{\SigmaOp}_{\epsilon}$ \textit{without sparsity and low rank}.
Recent techniques for ``quantum-inspired'' classical algorithms~\cite{10.1145/3313276.3316310} are not applicable either, since the full-rank operator $\hat{\SigmaOp}_{\epsilon}$ does not have a low-rank approximation.
Remarkably, our technique does not directly use the conventional ways that require sparsity or low rank, yet implements the block encoding of $\hat{\SigmaOp}_{\epsilon}$ efficiently.

Our significant technical contribution is to overcome the above difficulty by exploiting quantum Fourier transform (QFT) for efficient implementation of the block encoding of $\hat{\SigmaOp}_{\epsilon}$.
In our algorithm, QFTs are used for implementing the block encoding of $\SigmaOp_\epsilon$ and for applying $\F_D^\dag$ in preparing $\Ket{\Psi}$ in~\eqref{eq:psi}.
The sparse and low-rank assumptions can be avoided because we explicitly decompose the (non-sparse and full-rank) operator $\SigmaOp_\epsilon$ in~\eqref{eq:psi} into addition and the multiplication of diagonal (i.e., sparse) operators and QFTs\@.
We could efficiently implement $\SigmaOp_\epsilon$ by addition and multiplication of block encodings of these diagonal operators and QFTs, but presentation of these additions and multiplications may become complicated since we have multiple block encodings to be combined. For simplicity of the presentation, we use the block encoding of the POVM operator~\cite{G1} at the technical level to represent how to combine all the block encodings and QFTs as one circuit, as shown in Figs.~1 and~2 of Supplemental Material.
In particular, by the perfect reconstruction~\eqref{eq:perfect_reconstruction}, we decompose $\hat{\SigmaOp}_{\epsilon}$ into diagonal operators $\sqrt{(\nicefrac{1}{Q_{\max}^{(\tau)}})\Q^{(\tau)}}$, $\sqrt{\hat{\q}^{\left(\rho\right)}}$ (whose block encodings are efficiently implementable) and unitary operators $\F_D$, $\F_D^\dag$ representing $D$-dimensional discrete Fourier transform and its inverse.
The QFT provides a quantum circuit implementing $\F_D$ (and $\F_D^\dag$) with precision $\Delta$ within time $O(D\log({G})\log(\nicefrac{\log G}{\Delta}))$~\cite{C5}.
We combine these implementations to obtain a quantum circuit that efficiently implements the block encoding of $\hat{\SigmaOp}_{\epsilon}$.
The QSVT of our block encoding of $\hat{\SigmaOp}_{\epsilon}$ yields a block encoding of $\hat{\SigmaOp}_{\epsilon}^{-\frac{1}{2}}$ with precision $\Delta$, using the block encoding of $\hat{\SigmaOp}_{\epsilon}$ repeatedly $\widetilde{O}((\nicefrac{Q_{\max}^{(\tau)}}{{\epsilon}})\polylog(\nicefrac{1}{\Delta}))$ times~\cite{G1}, where the factor $\nicefrac{Q_{\max}^{(\tau)}}{{\epsilon}}$ is obtained from the condition number of $\hat{\SigmaOp}_{\epsilon}$, and $\widetilde{O}$ may ignore poly-logarithmic factors.
Using these techniques, we achieve the sampling from $Q_{\epsilon}^\ast(v_{{G}})P^{(\tau)}(v_{{G}})$ within a linear runtime in data dimension $D$ under the assumption that we show in the next section.
See Algorithm~1 in Supplementary Material for detail.

\subsection{\label{sec:assumption_kernel}Runtime analysis of learning with optimized random features}

We bound the runtime of learning with optimized random features achieved by our quantum algorithm.
In our runtime analysis, we use the following model of accessing given examples of data.
Abstracting a device implementing random access memory (RAM) in classical computation,
we assume access to the $n$th example of data via oracle functions
$\mathcal{O}_{\tilde{x}}(n)=\tilde{x}_n$ and $\mathcal{O}_y(n)=y_n$ mapping $n\in\left\{0,\ldots,N-1\right\}$ to the examples.
Analogously to sampling $\tilde{x}\in\tilde{\mathcal{X}}$ with probability $\hat{q}(\tilde{x})$,
we allow a quantum computer to use a quantum oracle (i.e., a unitary) $\mathcal{O}_\rho$ to set a quantum register $\mathcal{H}^X$ in a quantum state
$\sum_{\tilde{x}\in\tilde{\mathcal{X}}}\sqrt{\hat{q}^{(\rho)}(\tilde{x})}\Ket{\tilde{x}}$
so that we can sample $\tilde{x}$ with probability $\hat{q}\left(\tilde{x}\right)$ by a measurement of this state in the computational basis $\{\Ket{\tilde{x}}\}$.
This input model $\mathcal{O}_\rho$ is \textit{implementable feasibly  and efficiently using techniques in Refs.~\cite{G3,kerenidis_et_al:LIPIcs:2017:8154}} combined with a quantum device called quantum RAM (QRAM)~\cite{PhysRevLett.100.160501,PhysRevA.78.052310}, as discussed in Supplemental Material.
The oracle $\mathcal{O}_\rho$ is the only black box in our quantum algorithm; putting effort to make our algorithm explicit, we avoid any other use of QRAM\@.
Note that the time required for accessing data is indeed a matter of computational architecture and data structure,
and the focus of this paper is algorithms rather than architectures.
The runtime for each query to $\mathcal{O}_{\tilde{x}}$, $\mathcal{O}_y$, and $\mathcal{O}_\rho$ is denoted by $T_{\tilde{x}}$, $T_y$, and $T_\rho$, respectively.
The runtime of our algorithm does not explicitly depend on the number $N$ of given examples except that the required runtime $T_{\tilde{x}}$, $T_y$, and $T_\rho$ for accessing the data may depend on $N$, which we expect to be $O(1)$ or $O(\polylog(N))$.

Our algorithm can use any translation-invariant kernel $\tilde{k}$ given in the form of~\eqref{eq:perfect_reconstruction},
where $Q^{(\tau)}(v_{{G}})$ can be given by any function efficiently computable by classical computation in a short time denoted by $T_\tau=O(\poly(D))$,
and the maximum ${Q_{\max}^{(\tau)}}$ of $Q^{(\tau)}(v_{{G}})$ in~\eqref{eq:psi} is also assumed to be given.
We assume bounds $\tilde{k}(0,0)=\Omega\left(k(0,0)\right)=\Omega(1)$ and $Q_{\max}^{(\tau)}=O\left(\poly\left(D\right)\right)$,
which mean that the parameters of the kernel function are adjusted appropriately, so that $\tilde{k}(0,0)$ can reasonably approximate $k(0,0)=\int_\mathcal{V}d\tau(v)=1$, and $Q_{\max}^{(\tau)}\left(=\Omega(1)\right)$ may not be too large (e.g., not exponentially large) as $D$ gets large.
Remarkably, representative choices of kernels, such as the Gaussian kernel and the Laplacian kernel in Table~\ref{table:kernel}, satisfy our assumptions in a reasonable parameter region,\footnote{%
For the kernels in Table~\ref{table:kernel}, $Q^{(\tau)}$ is a product of $D$ special functions, computable in time $T_\tau=O(D)$ if each special function is computable in a constant time.
It is immediate to give
$Q_{\max}^{(\tau)}=Q^{(\tau)}\left(0\right)$.
We have $\tilde{k}(0,0)\geqq 1=\Omega\left(1\right)$ for these kernels.
We can also fulfill $Q_{\max}^{(\tau)}=O\left(\poly\left(D\right)\right)$ by reducing the parameter $\gamma$ of the kernels in Table~\ref{table:kernel} as $D$ increases (the reduction of $\gamma$ enlarges the class of learnable functions).
}
and not only these kernels,
we can use any kernel satisfying our assumptions.
Our algorithm \textit{does not impose sparsity or low rank} on $\kernel$ for the kernel and $\hat{\q}^{\left(\rho\right)}$ for the data distribution.
Note that the requirement~\eqref{eq:effective_dim_bound} of the upper bound of the degree of freedom $d(\epsilon)$ \textit{does not imply low rank} of $\kernel$ and $\hat{\q}^{\left(\rho\right)}$ while low-rank $\kernel$ or low-rank $\hat{\q}^{\left(\rho\right)}$ would conversely lead to an upper bound of $d(\epsilon)$.
Hence, our algorithm is widely applicable compared to existing QML algorithms shown in Sec.~\ref{sec:intro}.

We prove that our quantum algorithm achieves the following runtime $T_1$.
Significantly, $T_1$ is as fast as linear in $D$ whereas no existing classical algorithm achieves this sampling in sub-exponential time.
Note that the precision factor $\polylog(\nicefrac{1}{\Delta})$ in $T_1$ of the following theorem is ignorable in practice.\footnote{%
E.g., inner product of $D$-dimensional real vectors is calculable in time $O(D\polylog(\nicefrac{1}{\Delta}))$ with precision $\Delta$ using $O(\log(\nicefrac{1}{\Delta}))$-bit fixed-point number representation, but the factor $\polylog(\nicefrac{1}{\Delta})$ is practically ignored.
}

\begin{theorem}\label{thm:sampling}
  Given $D$-dimensional data discretized by $G>0$,
  for any learning accuracy ${\epsilon}>0$
  and any sampling precision $\Delta>0$,
  the runtime $T_1$ of our quantum algorithm for sampling each optimized random feature
   $v_{{G}}\in\mathcal{V}_{{G}}$ from a distribution $Q(v_G)P^{(\tau)}(v_G)$ close to the optimized distribution $Q_\epsilon^\ast(v_G)P^{(\tau)}(v_G)$ with precision
  $\sum_{v_{{G}}\in\mathcal{V}_{{G}}}{|Q(v_G)P^{(\tau)}(v_G)-Q_\epsilon^\ast(v_G)P^{(\tau)}(v_G)|}\leqq\Delta$ is
  \begin{align*}
    T_1 &=O(D\log({G})\log\log({G})+T_\rho+T_\tau)\times\widetilde{O}((\nicefrac{Q_{\max}^{(\tau)}}{\epsilon})\polylog(\nicefrac{1}{\Delta})).
  \end{align*}
\end{theorem}

Furthermore, using $M$ optimized random features $v_0,\ldots,v_{M-1}$ sampled efficiently by this quantum algorithm,
we construct an algorithm achieving the learning as a whole (See Algorithm~2 in Supplementary Material), where this $M$ is to be chosen appropriately to satisfy~\eqref{eq:M_bound}.
To achieve the learning,
we need to obtain coefficients $\alpha_0,\ldots,\alpha_{M-1}$ of $\hat{f}_{M,v_m,\alpha_m}=\sum_{m=0}^{M-1}\alpha_m\varphi(v_m,\cdot)\approx f$ that reduce the generalization error to $O(\epsilon)$.
To perform regression for obtaining $\alpha_0,\ldots,\alpha_{M-1}$,
we use stochastic gradient descent (SGD)~\cite{pmlr-v99-jain19a} (Algorithm~3 in Supplementary Material) as in the common practice of machine learning.
Note that the performance of SGD with random features is extensively studied in Ref.~\cite{carratino2018learning}, but our contribution is to clarify its \textit{runtime} by evaluating the runtime per iteration of SGD explicitly.
As discussed in Sec.~\ref{sec:assumptions}, we aim to clarify the runtime of the learning in the large-scale limit;
in particular, we assume that the number $N$ of given examples of data is sufficiently large $N>T$, where $T$ is the number of iterations in the SGD\@.
Then, the sequence of given examples of data $\left(\tilde{x}_0,y_0\right),\left(\tilde{x}_1,y_1\right),\ldots$ provides observations of an IID random variable, and SGD converges to the minimum of the generalization error.
Combining our quantum algorithm with the SGD, we achieve the following runtime $T_2$ of supervised learning with optimized random features, which is as fast as linear in $M$ and $D$, i.e., $T_2=O\left(MD\right)$.
Significantly, the required number $M$ of features for our algorithm using the optimized features is expected to be nearly \textit{minimal}, whereas it has been computationally hard in practice to use the optimized features in classical computation.

\begin{theorem}\label{thm:complexity_all}
  (Informal)
  Overall runtime $T_2$ of learning with optimized random features is
  \begin{equation*}
    T_2=O(MT_1)+O((MD+T_{\tilde{x}}+T_y)\times(\nicefrac{1}{\epsilon^2})),
  \end{equation*}
  where $T_1$ appears in Theorem~\ref{thm:sampling}, the first term is the runtime of sampling $M$ optimized random features by our quantum algorithm, and the second term is the runtime of the SGD\@.
\end{theorem}

\section{Conclusion}
We have constructed a quantum algorithm for sampling an \textit{optimized random feature} within a linear time $O(D)$ in data dimension $D$,
achieving an exponential speedup in $D$ compared to the existing classical algorithm~\cite{B1,NIPS2018_7598,Shahrampour2019} for this sampling task.
Combining $M$ features sampled by this quantum algorithm with stochastic gradient descent, we can achieve supervised learning in time $O(MD)$ without canceling out the exponential speedup, where this $M$ is expected to be nearly minimal since we use the optimized random features. As for future work, it is open to prove hardness of sampling an optimized random feature for \textit{any} possible classical algorithm under complexity-theoretical assumptions.
It is also interesting to investigate whether we can reduce the runtime to $O(M\log D)$, as in Refs.~\cite{L1,NIPS2016_6246} but using the optimized random features to achieve minimal $M$.
Since our quantum algorithm does not impose sparsity or low-rank assumptions,
our results open a route to a widely applicable framework of kernel-based quantum machine learning with an exponential speedup.

\section*{Broader Impact}

Quantum computation has recently been attracting growing attentions owing to its potential for achieving computational speedups compared to any conventional classical computation that runs on existing computers, opening the new field of accelerating machine learning tasks via quantum computation: \textit{quantum machine learning}.
To attain a large quantum speedup, however,
existing algorithms for quantum machine learning require extreme assumptions on sparsity and low rank of matrices used in the algorithms, which limit applicability of the quantum computation to machine learning tasks.
In contrast, the novelty of this research is to achieve an exponential speedup in quantum machine learning without the sparsity and low-rank assumptions, broadening the applicability of quantum machine learning.

Advantageously, our quantum algorithm eliminates the computational bottleneck faced by a class of existing classical algorithms for scaling up kernel-based learning algorithms by means of random features.
In particular, using this quantum algorithm, we can achieve the learning with the nearly \textit{optimal} number of features, whereas this optimization has been hard to realize due to the bottleneck in the existing classical algorithms.
A drawback of our quantum algorithm may arise from the fact that we use powerful quantum subroutines for achieving the large speedup, and these subroutines are hard to implement on existing or near-term quantum devices that cannot achieve universal quantum computation due to noise.
At the same time, these subroutines make our quantum algorithm hard to simulate by classical computation, from which stems the computational advantage of our quantum algorithm over the existing classical algorithms.
Thus, our results open a route to a widely applicable framework of kernel-based quantum machine learning with an exponential speedup, leading to a promising candidate of ``killer applications'' of universal quantum computers.

\begin{ack}
This work was supported by CREST (Japan Science and Technology Agency) JPMJCR1671, Cross-ministerial Strategic Innovation Promotion Program (SIP) (Council for Science, Technologyand Innovation (CSTI)), JSPS Overseas Research Fellowships,
a Cambridge-India Ramanujan scholarship from the Cambridge Trust and the SERB (Govt.\ of India),
and JSPS KAKENHI 18K18113.
\end{ack}

\newpage
\appendix
\setcounter{theorem}{0}

\section*{Supplementary Material}

In Supplementary Material, after summarizing basic notions of quantum computation, we provide proofs of theorems and propositions mentioned in the main text.
In Sec.~\ref{sec:quantum_suplementary}, the basic notions of quantum computation are summarized.
In Sec.~\ref{sec:oracles},
the feasibility of implementing a quantum oracle that we use in our quantum algorithm is summarized.
In Sec.~\ref{sec:1},
we show Proposition~\ref{sprp:perfect_reconstruction} on the perfect reconstruction of the kernel, which is a crucial technique in our quantum algorithm.
In Sec.~\ref{sec:2},
we show Proposition~\ref{sprp:state} on a quantum state that we use in  our quantum algorithm for sampling an optimized random feature.
In Sec.~\ref{sec:3},
we show our quantum algorithm (Algorithm~\ref{salg:random_feature_revised}) for sampling the optimized random feature,
and prove Theorem~\ref{sthm:sampling} on the runtime of Algorithm~\ref{salg:random_feature_revised}.
In Sec.~\ref{sec:4},
we show the overall algorithm (Algorithm~\ref{salg:data_approximation}) for learning with the optimized random features by combining Algorithm~\ref{salg:random_feature_revised} with stochastic gradient descent (Algorithm~\ref{salg:sgd}),
and prove Theorem~\ref{sthm:runtime_all} on the runtime of Algorithm~\ref{salg:data_approximation}.
Note that lemmas that we show for the runtime analysis of our quantum algorithm are presented in Sec.~\ref{sec:3}, and the proofs in the other sections do not require these lemmas on quantum computation.
The notations used in Supplementary Material is the same as those in the main text.

\section{\label{sec:quantum_suplementary}Quantum computation}

In this section, we summarize basic notions of quantum computation, referring to Refs.~\cite{N4,arXiv:1907.09415} for more detail.

Analogously to a bit $\{0,1\}$ in classical computation,
the unit of quantum computation is a quantum bit (qubit), mathematically represented by $\mathbb{C}^2$, i.e., a $2$-dimensional complex Hilbert space.
A fixed orthonormal basis of a qubit $\mathbb{C}^2$ is denoted by $\left\{\Ket{0}\coloneqq\left(\begin{smallmatrix}1\\0\end{smallmatrix}\right),\Ket{1}\coloneqq\left(\begin{smallmatrix}0\\1\end{smallmatrix}\right)\right\}$.
Similarly to a bit taking a state $b\in\{0,1\}$, a qubit takes a quantum state $\Ket{\psi}=\alpha_0\Ket{0}+\alpha_1\Ket{1}=\left(\begin{smallmatrix}\alpha_0\\\alpha_1\end{smallmatrix}\right)\in\mathbb{C}^2$.
While a register of $m$ bits takes values in ${\left\{0,1\right\}}^m$, a quantum register of $m$ qubits is represented by the tensor-product space ${\left(\mathbb{C}^2\right)}^{\otimes m}\cong\mathbb{C}^{2^m}$, i.e., a $2^m$-dimensional Hilbert space.
We may use $=$ rather than $\cong$ to represent isomorphism for brevity.
We let $\mathcal{H}$ denote a finite-dimensional Hilbert space representing a quantum register; that is, an $m$-qubit register is $\mathcal{H}=\mathbb{C}^{2^m}$.
A fixed orthonormal basis $\{\Ket{x}: x\in\{0,\ldots,2^m-1\}\}$ labeled by $m$-bit strings, or the corresponding integers, is called the \textit{computational basis} of $\mathcal{H}$.
A state of $\mathcal{H}$ can be denoted by $\Ket{\psi} = \sum_{x=0}^{2^m-1}\alpha_x\Ket{x}\in\mathcal{H}$.
Any quantum state $\Ket{\psi}$ requires an $L_2$ normalization condition $\left\|\Ket{\psi}\right\|_2=1$, and for any $\theta\in\mathbb{R}$, $\Ket{\psi}$ is identified with $\mathrm{e}^{\mathrm{i}\theta}\Ket{\psi}$.

In the bra-ket notation,
the conjugate transpose of the column vector $\Ket{\psi}$ is a row vector denoted by $\Bra{\psi}$, where $\Bra{\psi}$ and $\Ket{\psi}$ may be called a bra and a ket, respectively. The inner product of $\Ket{\psi}$ and $\Ket{\phi}$ is denoted by $\Braket{\psi|\phi}$, while their outer product $\Ket{\psi}\Bra{\phi}$ is a matrix.
The conjugate transpose of an operator $\A$ is denoted by $\A^\dag$, and the transpose of $\A$ with respect to the computational basis is denoted by $\A^\mathrm{T}$.

A measurement of a quantum state $\Ket{\psi}$ is a sampling process that returns a randomly chosen bit string from the quantum state.
An $m$-qubit state $\Ket{\psi}=\sum_{x=0}^{2^m-1}\alpha_x\Ket{x}$ is said to be in a superposition of the basis states $\Ket{x}$s.
A measurement of $\Ket{\psi}$ in the computational basis $\{\Ket{x}\}$ provides a random $m$-bit string $x\in{\{0,1\}}^m$ as outcome, with probability $p(x)={|\alpha_x|}^2$.
After the measurement, the state changes from $\Ket{\psi}$ to $\Ket{x}$ corresponding to the obtained outcome $x$, and loses the randomness in $\Ket{\psi}$; that is, to iterate the same sampling as this measurement, we need to prepare $\Ket{\psi}$ repeatedly for each iteration.
For two registers $\mathcal{H}^A\otimes\mathcal{H}^B$ and their state $\Ket{\phi}^{AB}=\sum_{x,x}\alpha_{x,x^\prime}\Ket{x}^A\otimes\Ket{x^\prime}^B\in\mathcal{H}^A\otimes\mathcal{H}^B$, a measurement of the register $\mathcal{H}^B$ for $\Ket{\phi}^{AB}$ in the computational basis $\{\Ket{x^\prime}^B\}$ of $\mathcal{H}^B$ yields an outcome $x^\prime$ with probability $p(x^\prime)=\sum_{x}p(x,x^\prime)$, where $p(x,x^\prime)={|\alpha_{x,x^\prime}|}^2$.
The superscripts of a state or an operator represent which register the state or the operator belongs to, while we may omit the superscripts if it is clear from the context.

A quantum algorithm starts by initializing $m$ qubits in a fixed state $\Ket{0}^{\otimes m}$, which we may write as $\Ket{0}$ if $m$ is clear from the context.
Then, we apply a $2^m$-dimensional unitary operator $\U$ to $\Ket{0}^{\otimes m}$, to prepare a state $\U\Ket{0}^{\otimes m}$.
Finally, a measurement of $\U\Ket{0}^{\otimes n}$ is performed to sample an $m$-bit string from a probability distribution given by $\U\Ket{0}^{\otimes m}$.
Analogously to classical logic-gate circuits, $\U$ is represented by a quantum circuit composed of sequential applications of unitaries acting at most two qubits at a time.
Each of these unitaries is called an elementary quantum gate.
The runtime of a quantum algorithm represented by a quantum circuit is determined by the number of applications of elementary quantum gates in the circuit.

With techniques shown in Refs.~\cite{Subramanian_2019,chakraborty2018,G1},
non-unitary operators can also be used in quantum computation.
In particular,
to apply a non-unitary operator $\A$ in quantum computation, we use the technique of \textit{block encoding}~\cite{G1}, as summarized in the following.
A block encoding of $\A$ is a unitary operator $\U=\left(\begin{smallmatrix}\A&\cdot\\\cdot&\cdot\end{smallmatrix}\right)$ that encodes $\A$ in its left-top (or $\Ket{0}\Bra{0}$) subspace (up to numerical precision).
Note that we have
\begin{equation}
  \label{seq:block_encoding}
  \U=\left(\begin{smallmatrix}\mathbf{A}&\mathbf{B}\\\mathbf{C}&\mathbf{D}\end{smallmatrix}\right)=\Ket{0}\Bra{0}\otimes \mathbf{A}+\Ket{0}\Bra{1}\otimes \mathbf{B}+\Ket{1}\Bra{0}\otimes \mathbf{C}+\Ket{1}\Bra{1}\otimes \mathbf{D},
\end{equation}
if $\mathbf{A}$, $\mathbf{B}$, $\mathbf{C}$, and $\mathbf{D}$ are on the Hilbert space of the same dimension.
Consider a state $\Ket{0}\otimes\Ket{\psi}=\left(\begin{smallmatrix}\Ket{\psi}\\ \mathbf{0}\end{smallmatrix}\right)$ in the top-left (or $\Ket{0}\Bra{0}$) subspace of $\U$, where $\mathbf{0}$ is a zero column vector, and $\Ket{0}\in\mathbb{C}^d$ for some $d$.
Applying $\U$ to the state $\Ket{0}\otimes\Ket{\psi}$, we would obtain
\begin{equation}
  \U\left(\Ket{0}\otimes\Ket{\psi}\right)=\sqrt{p}\Ket{0}\otimes \frac{\A\Ket{\psi}}{\left\|\A\Ket{\psi}\right\|_2}+\sqrt{1-p}\Ket{\perp},
\end{equation}
where $p=\left\|\A\Ket{\psi}\right\|_2^2$, and $\Ket{\perp}$ is a state of no interest satisfying $(\Ket{0}\Bra{0}\otimes\mathbbm{1})\Ket{\perp}$.
Then, we can prepare the state to which $\A$ is applied, i.e.,
\begin{equation}
  \frac{\A\Ket{\psi}}{\left\|\A\Ket{\psi}\right\|_2}
\end{equation}
using this process for preparing $\U\left(\Ket{0}\otimes\Ket{\psi}\right)$ and its inverse process repeatedly $O(\frac{1}{\sqrt{p}})$ times, by means of amplitude amplification~\cite{BrassardHoyer}.
Note that given a quantum circuit, its inverse can be implemented by replacing each gate in the circuit with its inverse gate; that is, the circuit and its inverse circuit have the same runtime since they are composed of the same number of gates.
In Sec.~\ref{sec:3},
we will use the following more precise definition of block encoding to take the precision $\Delta$ into account.
For any operator $\A$ on $s$ qubits, i.e., on $\mathbb{C}^{2^s}$, a unitary operator $\U$ on $(s+a)$ qubits, i.e., on $\mathbb{C}^{2^{s+a}}$, is called an $(\alpha,a,\Delta)$-block encoding of $\A$ if it holds that
\begin{equation}
  \label{seq:block_encoding_definition}
  \left\|\A-\alpha\left(\mathbbm{1}\otimes\Bra{0}^{\otimes a}\right)\U\left(\mathbbm{1}\otimes\Ket{0}^{\otimes a}\right)\right\|_\infty\leqq \Delta,
\end{equation}
where $\|\cdot\|_\infty$ is the operator norm.
Note that since any unitary operator $\U$ satisfies $\left\|\U\right\|_\infty\leqq 1$, it is necessary that $\left\|\A\right\|_\infty\leqq\alpha+\Delta$.

\section{\label{sec:oracles}Feasibility of implementing quantum oracle}

In this section, we summarize the feasibility of implementing a quantum oracle that we use in our quantum algorithm.

The quantum oracles are mathematically represented by unitary operators.
As shown in the main text,
to access given examples of data in our quantum algorithm,
we use a quantum oracle $\mathcal{O}_\rho$ acting as
\begin{equation}
\label{seq:oracle_rho_definition}
\mathcal{O}_\rho(\Ket{0})=\sum_{\tilde{x}\in\tilde{\mathcal{X}}}\sqrt{\hat{q}^{(\rho)}(\tilde{x})}\Ket{\tilde{x}}
=\sqrt{\hat{\q}^{(\rho)}}\sum_{\tilde{x}\in\tilde{\mathcal{X}}}\Ket{\tilde{x}},
\end{equation}
where we write
\begin{equation}
    \hat{\q}^{(\rho)}=\sum_{\tilde{x}\in\tilde{\mathcal{X}}}\hat{q}^{(\rho)}\left(\tilde{x}\right)\Ket{\tilde{x}}\Bra{\tilde{x}}.
\end{equation}

We can efficiently implement the quantum oracle $\mathcal{O}_\rho$ with an acceptable preprocessing overhead
using the $N$ given examples of input data $\tilde{x}_0,\ldots,\tilde{x}_{N-1}$.
From these examples,
we can prepare a data structure proposed in Ref.~\cite{kerenidis_et_al:LIPIcs:2017:8154} in $O(N{(D\log{G})}^2)$ time using $O(N{(D\log{G})}^2)$ bits of memory,
while collecting and storing the $N$ data points
requires at least $\Theta(ND\log{G})$ time and $\Theta(ND\log{G})$ bits of memory.
Note that this data structure is also used in ``quantum-inspired'' classical algorithms~\cite{10.1145/3313276.3316310,arXiv:1910.05699,arXiv:1910.06151}.
Then, we can implement $\mathcal{O}_\rho$ by a quantum circuit combined with a quantum random access memory (QRAM)~\cite{PhysRevA.78.052310,PhysRevLett.100.160501}, which can load data from this data structure into qubits in quantum superposition (i.e.\ linear combinations of quantum states).
With $T_\mathrm{Q}$ denoting runtime of this QRAM per query,
it is known that this implementation of $\mathcal{O}_\rho$ with precision $\Delta$ has runtime
\begin{equation}
T_\rho=O(D\log({G})\polylog(\nicefrac{1}{\Delta})\times T_\mathrm{Q})
\end{equation}
per query~\cite{kerenidis_et_al:LIPIcs:2017:8154,G3}.
The runtime $T_\mathrm{Q}$ of this QRAM may scale poly-logarithmically in $N$ depending on how we implement the QRAM, but such an implementation suffices to meet our expectation in the main text that $T_\rho$ should be $O(1)$ or $O(\polylog(N))$ as $N$ increases.
Note that the inverse $\mathcal{O}_\rho^\dag$ of $\mathcal{O}_\rho$ has the same runtime $T_\rho$ since $\mathcal{O}_\rho^\dag$ can be implemented by replacing each quantum gate in the circuit for $\mathcal{O}_\rho$ with its inverse.

Thus, if both the quantum computer and the QRAM are available, we can implement $\mathcal{O}_\rho$ feasibly and efficiently.
Similarly to the quantum computer assumed to be available in this paper,
QRAM is actively under development towards its physical realization; e.g., see Refs.~\cite{jiang2019experimental,PhysRevLett.123.250501} on recent progress towards realizing QRAM\@.
The use of QRAM is a common assumption in quantum machine learning (QML) especially to deal with a large amount of data;
however, even with QRAM, achieving quantum speedup is nontrivial.
Note that we do not include the time for collecting the data or preparing the above data structure in runtime of our learning algorithm, but even if we took them into account, an exponential speedup from $O(\exp(D))$ to $O(\poly(D))$ would not be canceled out.
Since we exploit $\mathcal{O}_\rho$ for constructing a widely applicable QML framework achieving the exponential speedup without sparsity and low-rank assumptions, our results motivate further technological development towards realizing the QRAM as well as the quantum computer.

\section{\label{sec:1}Perfect reconstruction of kernel}

In this section, we show the following perfect reconstruction of the kernel that we use in our quantum algorithm.

\begin{proposition}[\label{sprp:perfect_reconstruction}Perfect reconstruction of kernel]
  Given any periodic translation-invariant kernel $\tilde{k}$,
  we exactly have for each $\tilde{x}^\prime,\tilde{x}\in\tilde{\mathcal{X}}$
  \begin{align*}
    \tilde{k}\left(\tilde{x}^\prime,\tilde{x}\right)&=\sum_{v_G\in\mathcal{V}_G}\frac{Q^{(\tau)}(v_G)}{G^D}\overline{\varphi(v_G,\tilde{x}^\prime)}\varphi(v_G,\tilde{x})\nonumber\\
    &=\Bra{\tilde{x}^\prime}\F_D^\dag \Q^{(\tau)} \F_D\Ket{\tilde{x}}=\Bra{\tilde{x}^\prime}\F_D \Q^{(\tau)} \F_D^\dag\Ket{\tilde{x}}.
  \end{align*}
\end{proposition}

\begin{proof}
  To show the perfect reconstruction of the kernel $\tilde{k}$,
   we crucially use the assumption given in the main text that the data domain is finite due to the discretized representation
  \begin{equation}
    \tilde{\mathcal{X}}={\left\{0,1,\ldots,{G}-1\right\}}^D.
  \end{equation}
As summarized in the main text,
recall that we approximate a translation-invariant (but not necessarily periodic) kernel $k\left(x^\prime,x\right)$ by
\begin{equation}
  \label{seq:kernel_approximation}
  \tilde{k}\left(x^\prime,x\right)=\sum_{n\in\mathbb{Z}^D}k\left(x^\prime,x+Gn\right).
\end{equation}
To represent the translation-invariant kernel functions,
we may write
\begin{align}
  k_\mathrm{TI}\left(x^\prime-x\right)&\coloneqq k\left(x^\prime,x\right),\\
  \tilde{k}_\mathrm{TI}\left(x^\prime-x\right)&\coloneqq \tilde{k}\left(x^\prime,x\right).
\end{align}
The function $\tilde{k}$ is periodic by definition; in particular, we have for any $n^\prime\in\mathbb{Z}^D$
\begin{equation}
  \label{seq:periodicity}
  \tilde{k}\left(x^\prime,x\right)=\tilde{k}\left(x^\prime+Gn^\prime,x\right)=\tilde{k}\left(x^\prime,x+Gn^\prime\right)=\tilde{k}_\mathrm{TI}(x^\prime-x+Gn^\prime).
\end{equation}
Recall that the translation-invariant kernel $k:\mathcal{X}\times\mathcal{X}\to\mathbb{R}$ can be written as
\begin{equation}
  \label{seq:kernel_in_expectation}
  k\left(x^\prime,x\right)=\int_\mathcal{V}d\tau\left(v\right)\overline{\varphi\left(v,x^\prime\right)}\varphi\left(v,x\right),
\end{equation}
where $\varphi(v,x)\coloneqq\mathrm{e}^{-2\pi\mathrm{i}v\cdot x}$, and $d\tau$ is given by the Fourier transform of the kernel, in particular,~\cite{R2}
\begin{equation}
  \label{seq:dtau}
  d\tau(v)=q^{(\tau)}(v)dv=\left[\int_\mathcal{X}dx\,\mathrm{e}^{-2\pi\mathrm{i}v\cdot x}k_\mathrm{TI}\left(x\right)\right] dv.
\end{equation}
Similarly to~\eqref{seq:dtau}, our proof will expand $\tilde{k}$ using the Fourier transform.

To expand $\tilde{k}$,
we first consider the case of $D=1$, and will later consider $D\geqq 1$ in general.
In the case of $D=1$,
Shannon's sampling theorem~\cite{shannon1949communication} in signal processing~\cite{proakis2001digital} shows that we can perfectly reconstruct the kernel function $\tilde{k}_\mathrm{TI}$ on a \textit{continuous} domain $\left[-\frac{G}{2},\frac{G}{2}\right]$ from \textit{discrete} frequencies of its Fourier transform.
In the one-dimensional case, the Fourier transform of $\tilde{k}_\mathrm{TI}$ on $\left[-\frac{G}{2},\frac{G}{2}\right]$ is
\begin{equation}
  \int_{-\frac{G}{2}}^{\frac{G}{2}}dx\,\tilde{k}_\mathrm{TI}(x)\mathrm{e}^{-2\pi\mathrm{i}vx}=\int_{-\infty}^{\infty}dx\,k_\mathrm{TI}(x)\mathrm{e}^{-2\pi\mathrm{i}vx}=q^{(\tau)}(v).
\end{equation}
Then, for any $x\in\left[-\frac{G}{2},\frac{G}{2}\right]$, using the discrete frequencies $\tilde{v}\in\mathbb{Z}$ for $q^{(\tau)}\left(\tilde{v}\right)$,
we exactly obtain from the sampling theorem
\begin{equation}
  \label{seq:one_dim}
  \tilde{k}_\mathrm{TI}\left(x\right)=\frac{1}{G}\sum_{\tilde{v}=-\infty}^{\infty}q^{(\tau)}\left(\frac{\tilde{v}}{G}\right)\mathrm{e}^{2\pi\mathrm{i}\left(\frac{\tilde{v}}{G}\right)x}=\frac{1}{G}\sum_{\tilde{v}=-\infty}^{\infty}q^{(\tau)}\left(\frac{\tilde{v}}{G}\right)\mathrm{e}^{\frac{2\pi\mathrm{i}\tilde{v}x}{G}}.
\end{equation}
Due to the periodicity~\eqref{seq:periodicity} of $\tilde{k}_\mathrm{TI}$,~\eqref{seq:one_dim} indeed holds for any $x\in\mathbb{R}$.
In the same way, for any $D\geqq 1$, we have for any $x\in\mathbb{R}^D$
\begin{equation}
  \label{seq:D_dim}
  \tilde{k}_\mathrm{TI}\left(x\right)=\frac{1}{{G}^D}\sum_{\tilde{v}\in\mathbb{Z}^D}q^{(\tau)}\left(\frac{\tilde{v}}{G}\right)\mathrm{e}^{\frac{2\pi\mathrm{i}\tilde{v}\cdot x}{G}}.
\end{equation}

In addition, since $\tilde{\mathcal{X}}$ is a \textit{discrete} domain spaced at intervals $1$,
we can achieve the perfect reconstruction of the kernel $\tilde{k}_\mathrm{TI}$ on $\tilde{\mathcal{X}}$ by the $D$-dimensional discrete Fourier transform of $\tilde{k}_\mathrm{TI}$, using a \textit{finite} set of discrete frequencies for $q^{(\tau)}$.
In particular, for each $\tilde{v}\in\tilde{\mathcal{X}}$, the discrete Fourier transform of $\tilde{k}_\mathrm{TI}$ yields
\begin{align}
  \label{seq:discrete_fourier}
  \frac{1}{\sqrt{{G}^D}}\sum_{\tilde{x}\in\tilde{\mathcal{X}}}\tilde{k}_\mathrm{TI}\left(\tilde{x}\right)\mathrm{e}^{\frac{-2\pi\mathrm{i}\tilde{v}\cdot \tilde{x}}{G}}&=\frac{1}{\sqrt{{G}^D}}\sum_{\tilde{x}\in\tilde{\mathcal{X}}}\left(\frac{1}{{G}^D}\sum_{\tilde{v}^{\prime\prime}\in\mathbb{Z}^D}q^{(\tau)}\left(\frac{\tilde{v}^{\prime\prime}}{G}\right)\mathrm{e}^{\frac{2\pi\mathrm{i}\tilde{v}^{\prime\prime}\cdot \tilde{x}}{G}}\right)\mathrm{e}^{\frac{-2\pi\mathrm{i}\tilde{v}\cdot \tilde{x}}{G}}\nonumber\\
  &=\frac{1}{\sqrt{{G}^D}}\sum_{\tilde{v}^\prime\in\mathbb{Z}^D}q^{(\tau)}\left(\frac{\tilde{v}}{G}+\tilde{v}^\prime\right),
\end{align}
where the sum over $\tilde{x}$ in the first line is nonzero if $\tilde{v}^{\prime\prime}=\tilde{v}+G\tilde{v}^\prime$ for any $\tilde{v}^\prime\in\mathbb{Z}^D$.
Thus for the perfect reconstruction of the kernel $\tilde{k}$ on this domain $\tilde{\mathcal{X}}$, it suffices to use feature points $v_{G}=\frac{\tilde{v}}{G}$ for each $\tilde{v}\in\tilde{\mathcal{X}}$, which yields a finite set $\mathcal{V}_{G}$ of features
\begin{equation}
  \label{seq:V_G}
  v_{G}=\left(\begin{matrix}
      v_{G}^{\left(1\right)}\\
      \vdots\\
      v_{G}^{\left(D\right)}
  \end{matrix}\right)
      \in\mathcal{V}_{G}\coloneqq{\left\{0,\frac{1}{G},\ldots,1-\frac{1}{G}\right\}}^D.
\end{equation}
We use the one-to-one correspondence between $v_{G}\in\mathcal{V}_{G}$ and $\tilde{x}\in\tilde{\mathcal{X}}$ satisfying
\begin{equation}
  \label{seq:correspondence_v}
  v_{G}=\frac{\tilde{x}}{G},
\end{equation}
which we may also write using $\tilde{v}=\tilde{x}$ as
\begin{equation}
  v_{G}=\frac{\tilde{v}}{G}.
\end{equation}
In the same way as the main text, we let $Q^{(\tau)}:\mathcal{V}_{G}\to\mathbb{R}$ denote the function in~\eqref{seq:discrete_fourier}
\begin{equation}
  \label{seq:q_tau_g_function}
  Q^{(\tau)}\left(v_{G}\right)\coloneqq\sum_{\tilde{v}^\prime\in\mathbb{Z}^D}q^{(\tau)}\left(v_{G}+\tilde{v}^\prime\right).
\end{equation}

Therefore, from the $D$-dimensional discrete Fourier transform of~\eqref{seq:discrete_fourier},
we obtain the perfect reconstruction of the kernel $\tilde{k}_\mathrm{TI}$ on the domain $\tilde{\mathcal{X}}$ using the feature points in $\mathcal{V}_{G}$ and the function $Q^{(\tau)}$ as
\begin{align}
  \label{seq:perfect_reconstruction}
  &\tilde{k}\left(\tilde{x}^\prime,\tilde{x}\right)=\tilde{k}_\mathrm{TI}\left(\tilde{x}^\prime-\tilde{x}\right)\nonumber\\
  &=\frac{1}{\sqrt{{G}^D}}\sum_{\tilde{v}\in\tilde{\mathcal{X}}}\left(\frac{1}{\sqrt{{G}^D}}\sum_{\tilde{v}^\prime\in\mathbb{Z}^D}q^{(\tau)}\left(\frac{\tilde{v}}{G}+\tilde{v}^\prime\right)\right)\mathrm{e}^{\frac{2\pi\mathrm{i}\tilde{v}\cdot \left(\tilde{x}^\prime-\tilde{x}\right)}{G}}\nonumber\\
  &=\sum_{v_{G}\in\mathcal{V}_{G}}\frac{Q^{(\tau)}\left(v_{G}\right)}{{G}^{D}}\overline{\varphi\left(v_{G},\tilde{x}^\prime\right)}\varphi\left(v_{G},\tilde{x}\right),\quad\forall \tilde{x}^\prime,\tilde{x}\in\tilde{\mathcal{X}},
\end{align}
which shows the first equality in Proposition~\ref{sprp:perfect_reconstruction}.
Note that this equality also leads to a lower bound of $Q_{\max}^{(\tau)}$, that is, the maximum of $Q^{(\tau)}\left(v_{G}\right)$, as shown in Remark~\ref{rem:perfect} after this proof.

To show the second equality in Proposition~\ref{sprp:perfect_reconstruction}, recall that we write a diagonal operator corresponding to $Q^{(\tau)}\left(v_{G}\right)$ as
\begin{equation}
  \label{seq:q_tau_g}
  \Q^{(\tau)}\coloneqq\sum_{\tilde{v}\in\tilde{\mathcal{X}}}Q^{(\tau)}\left(\frac{\tilde{v}}{G}\right)\Ket{\tilde{v}}\Bra{\tilde{v}}.
\end{equation}
Note that we write $\Ket{\tilde{v}}=\Ket{\tilde{x}}$ for $\tilde{v}=\tilde{x}\in\tilde{\mathcal{X}}$ for clarity of the presentation.
In addition, let $\F$ denote a unitary operator representing (one-dimensional) discrete Fourier transform
\begin{equation}
  \label{seq:F}
  \F\coloneqq\sum_{\tilde{x}=0}^{{G}-1}\left(\frac{1}{\sqrt{G}}\sum_{\tilde{v}=0}^{{G}-1}\mathrm{e}^{-\frac{2\pi\mathrm{i}\tilde{v}\tilde{x}}{G}}\Ket{\tilde{v}}\right)\Bra{\tilde{x}},
\end{equation}
and $\F_{D}$ denote a unitary operator representing $D$-dimensional discrete Fourier transform
\begin{equation}
  \label{seq:F_D}
  \F_{D}\coloneqq\F^{\otimes D}=\sum_{\tilde{x}\in\tilde{\mathcal{X}}}\left(\frac{1}{\sqrt{{G}^D}}\sum_{\tilde{v}\in\tilde{\mathcal{X}}}\mathrm{e}^{-\frac{2\pi\mathrm{i}\tilde{v}\cdot\tilde{x}}{G}}\Ket{\tilde{v}}\right)\Bra{\tilde{x}}.
\end{equation}
The feature map can be written in terms of $\F_{D}$ as
\begin{equation}
  \label{seq:feature}
  \varphi\left(v_{G},\tilde{x}\right)=\mathrm{e}^{-2\pi\mathrm{i}v_{G}\cdot\tilde{x}}=\sqrt{{G}^D}\Braket{\tilde{v}|\F_{D}|\tilde{x}}=\sqrt{{G}^D}\Braket{\tilde{x}|\F_{D}|\tilde{v}},
\end{equation}
where $v_{G}=\frac{\tilde{v}}{G}$, and the last equality follows from the invariance of $\F_{D}$ under the transpose with respect to the computational basis.
From~\eqref{seq:perfect_reconstruction},~\eqref{seq:F_D}, and~\eqref{seq:feature}, by linear algebraic calculation, we obtain the conclusion for any $\tilde{x}^\prime,\tilde{x}\in\tilde{\mathcal{X}}$
\begin{equation}
  \label{seq:kernel_decomposition}
  \tilde{k}\left(\tilde{x}^\prime,\tilde{x}\right)=\Braket{\tilde{x}^\prime|\F_{D}^\dag \Q^{(\tau)} \F_{D}|\tilde{x}}=\Braket{\tilde{x}^\prime|\F_{D} \Q^{(\tau)} \F_{D}^\dag|\tilde{x}},
\end{equation}
where the last equality follows from the fact that the kernel function $\tilde{k}$ is symmetric and real, i.e., $\tilde{k}(x^\prime,x)=\tilde{k}(x,x^\prime)$ and $\overline{\tilde{k}(x^\prime,x)}=\tilde{k}(x^\prime,x)$.
\end{proof}

\begin{remark}[\label{rem:perfect}A lower bound of $Q_{\max}^{(\tau)}$]
  Equality~\eqref{seq:perfect_reconstruction} has the following implication on a lower bound of the maximum of $Q^{(\tau)}\left(v_{G}\right)$
  \begin{equation}
\label{seq:q_tau_g_max}
{Q_{\max}^{(\tau)}}=\max\left\{Q^{(\tau)}\left(v_{G}\right):v_{G}\in\mathcal{V}_{G}\right\}.
  \end{equation}
Recall that we let $P^{(\tau)}$ denote a probability mass function on $\mathcal{V}_G$ proportional to $Q^{(\tau)}$
\begin{equation}
  \label{seq:normalization_p_tau}
  P^{(\tau)}\left(v_G\right)\coloneqq\frac{Q^{(\tau)}\left(v_G\right)}{\sum_{v_G^\prime\in\mathcal{V}_G}Q^{(\tau)}\left(v_G^\prime\right)},
\end{equation}
which by definition satisfies the normalization condition
\begin{equation}
  \sum_{v_G\in\mathcal{V}_G}P^{(\tau)}\left(v_G\right)=1.
\end{equation}
We obtain from~\eqref{seq:perfect_reconstruction}
\begin{equation}
  \tilde{k}(0,0)=\sum_{v_{G}\in\mathcal{V}_{G}}\frac{Q^{(\tau)}\left(v_{G}\right)}{{G}^{D}},
\end{equation}
and hence, we can regard $\tilde{k}(0,0)$ as a normalization factor in
\begin{equation}
  P^{(\tau)}\left(v_G\right)=\frac{1}{\tilde{k}(0,0)} \frac{Q^{(\tau)}\left(v_G\right)}{G^D}.
\end{equation}
The normalization of $P^{(\tau)}$ yields a lower bound of $Q_{\max}^{(\tau)}$
\begin{equation}
  Q_{\max}^{(\tau)}=G^D\times\frac{Q_{\max}^{(\tau)}}{G^D}\geqq \sum_{v_G\in\mathcal{V}_G}\frac{Q^{(\tau)}\left(v_G\right)}{G^D}=\tilde{k}(0,0)\sum_{v_G\in\mathcal{V}_G}P^{(\tau)}\left(v_G\right)=\tilde{k}(0,0)=\Omega(1),
\end{equation}
where we use the assumption $\tilde{k}(0,0)=\Omega(k(0,0))=\Omega(1)$.
\end{remark}

\section{\label{sec:2}Quantum state for sampling an optimized random feature}

In this section, we show a quantum state that we use in our quantum algorithm for sampling an optimized random feature.
In particular, as shown in the main text, recall a quantum state on two quantum registers $\mathcal{H}^X\otimes\mathcal{H}^{X^\prime}$
\begin{equation}
  \label{seq:state_sampling}
  \Ket{\Psi}^{XX^\prime}\propto\sum_{\tilde{x}\in\tilde{\mathcal{X}}}\hat{\SigmaOp}_{\epsilon}^{-\frac{1}{2}}\Ket{\tilde{x}}^X\otimes \sqrt{\frac{1}{Q_{\max}^{\left(\tau\right)}}\Q^{\left(\tau\right)}}\F_D^\dag\sqrt{\hat{q}^{\left(\rho\right)}\left(\tilde{x}\right)}\Ket{\tilde{x}}^{X^\prime},
\end{equation}
where $X$ and $X^\prime$ have the same number of qubits.
Then, we show the following proposition.

\begin{proposition}[\label{sprp:state}Quantum state for sampling an optimized random feature]
  If we perform a measurement of the quantum register $X^\prime$ on the state $\Ket{\Psi}^{XX^\prime}$ defined as~\eqref{seq:state_sampling} in the computational basis $\{\Ket{\tilde{x}}^{X^\prime}:\tilde{x}\in\tilde{\mathcal{X}}\}$, then we obtain a measurement outcome $\tilde{x}$ with probability $Q_{\epsilon}^\ast\left(\frac{\tilde{x}}{{G}}\right)P^{\left(\tau\right)}\left(\frac{\tilde{x}}{{G}}\right)$.
\end{proposition}

\begin{proof}
  The proof is given by linear algebraic calculation.
  Note that the normalization $\left\|\Ket{\Psi}^{XX^\prime}\right\|_2=1$ of a quantum state always yields the normalization $\sum_{\tilde{x}^\prime\in\tilde{\mathcal{X}}}p\left(\tilde{x}^\prime\right)=1$ of a probability distribution obtained from the measurement of $\mathcal{H}^{X^\prime}$ in the computational basis $\left\{\Ket{\tilde{x}^\prime}^{X^\prime}\right\}$, and hence, we may omit the normalization constant in the following calculation for simplicity of the presentation.

  Recall the definition of the optimized probability distribution $Q_{\epsilon}^\ast\left(v_G\right)P^{\left(\tau\right)}\left(v_G\right)$
  \begin{equation}
    \label{seq:state_1}
    Q_{\epsilon}^\ast\left(v_G\right)P^{\left(\tau\right)}\left(v_G\right)=
    \frac{\Braket{\varphi\left(v_{G},\cdot\right) | \hat{\q}^{\left(\rho\right)}{\left({\hat{\SigmaOp}}+\epsilon\mathbbm{1}\right)}^{-1} | \varphi\left(v_{G},\cdot\right)}Q^{(\tau)}\left(v_{G}\right)}{\sum_{v_{G}^\prime\in\mathcal{V}_{G}}\Braket{\varphi\left(v_{G}^\prime,\cdot\right) | \hat{\q}^{\left(\rho\right)}{\left({\hat{\SigmaOp}}+\epsilon\mathbbm{1}\right)}^{-1} | \varphi\left(v_{G}^\prime,\cdot\right)}Q^{(\tau)}\left(v_{G}^\prime\right)},
  \end{equation}
  where we write
  \begin{align}
    \hat{\SigmaOp}&=\kernel \hat{\q}^{(\rho)},\\
    \kernel&=\sum_{\tilde{x}^\prime,\tilde{x}\in\tilde{\mathcal{X}}}\tilde{k}\left(\tilde{x}^\prime,\tilde{x}\right)\Ket{\tilde{x}^\prime}\Bra{\tilde{x}}.
  \end{align}
For $v_{G}=\frac{\tilde{x}}{G}$, it follows from~\eqref{seq:feature} that
\begin{equation}
  \Ket{\varphi\left(v_{G},\cdot\right)}=\sqrt{{G}^D}\F_{D}\Ket{\tilde{x}}.
\end{equation}
  Then, we have
  \begin{align}
    \label{seq:state_2}\eqref{seq:state_1}
    &=\frac{\Braket{\varphi\left(v_{G},\cdot\right) | \hat{\q}^{\left(\rho\right)}{\left({\hat{\SigmaOp}}+\epsilon\mathbbm{1}\right)}^{-1} | \varphi\left(v_{G},\cdot\right)}}{\sum_{v_{G}^\prime\in\mathcal{V}_{G}}\frac{Q^{(\tau)}\left(v_{G}^\prime\right)}{{G}^D}\Braket{\varphi\left(v_{G}^\prime,\cdot\right) | \hat{\q}^{\left(\rho\right)}{\left({\hat{\SigmaOp}}+\epsilon\mathbbm{1}\right)}^{-1} | \varphi\left(v_{G}^\prime,\cdot\right)}}\frac{Q^{(\tau)}\left(v_{G}\right)}{{G}^D} \nonumber\\
    &=\frac{\Braket{\tilde{x}|\F_{D}^\dag \hat{\q}^{\left(\rho\right)}{\left({\hat{\SigmaOp}}+\epsilon\mathbbm{1}\right)}^{-1} \F_{D} | \tilde{x}}}{\sum_{\tilde{x}^\prime\in\tilde{\mathcal{X}}}Q^{(\tau)}\left(\frac{\tilde{x}^\prime}{G}\right)\Braket{\tilde{x}^\prime|\F_{D}^\dag \hat{\q}^{\left(\rho\right)}{\left({\hat{\SigmaOp}}+\epsilon\mathbbm{1}\right)}^{-1} \F_{D} | \tilde{x}^\prime}}Q^{(\tau)}\left(\frac{\tilde{x}}{G}\right).
  \end{align}
  Then, using~\eqref{seq:q_tau_g}, we obtain
  \begin{equation}\eqref{seq:state_2}
      =\frac{\Braket{\tilde{x}|\sqrt{\Q^{(\tau)}}\F_{D}^\dag \hat{\q}^{\left(\rho\right)}{\left({\hat{\SigmaOp}}+\epsilon\mathbbm{1}\right)}^{-1} \F_{D} \sqrt{\Q^{(\tau)}}| \tilde{x}}}{\sum_{\tilde{x}^\prime\in\tilde{\mathcal{X}}}\Braket{\tilde{x}^\prime|\sqrt{\Q^{(\tau)}}\F_{D}^\dag \hat{\q}^{\left(\rho\right)}{\left({\hat{\SigmaOp}}+\epsilon\mathbbm{1}\right)}^{-1} \F_{D} \sqrt{\Q^{(\tau)}}| \tilde{x}^\prime}}.
  \end{equation}
  Therefore, it holds that
  \begin{align}
    \label{seq:hat_p_lambda}
    Q_{\epsilon}^\ast\left(\frac{\tilde{x}}{G}\right)P^{\left(\tau\right)}\left(\frac{\tilde{x}}{G}\right)
    &\propto{\Braket{\tilde{x}|\sqrt{\Q^{(\tau)}}\F_{D}^\dag \hat{\q}^{\left(\rho\right)}{\left({\hat{\SigmaOp}}+\epsilon\mathbbm{1}\right)}^{-1} \F_{D} \sqrt{\Q^{(\tau)}}| \tilde{x}}}\nonumber\\
    &={\Braket{\tilde{x}|\sqrt{\frac{1}{Q_{\max}^{(\tau)}}\Q^{(\tau)}}\F_{D}^\dag \hat{\q}^{\left(\rho\right)}{\left({\frac{1}{Q_{\max}^{(\tau)}}\hat{\SigmaOp}}+\frac{\epsilon}{Q_{\max}^{(\tau)}}\mathbbm{1}\right)}^{-1} \F_{D} \sqrt{\frac{1}{Q_{\max}^{(\tau)}}\Q^{(\tau)}}| \tilde{x}}}.
  \end{align}

  To simplify the form of~\eqref{seq:hat_p_lambda}, define a positive semidefinite operator on the support of $\hat{\q}^{\left(\rho\right)}$
\begin{align}
  \hat{\SigmaOp}_\epsilon^{\left(\rho\right)}&\coloneqq \sqrt{\hat{\q}^{\left(\rho\right)}}\left({\frac{1}{Q_{\max}^{(\tau)}}\hat{\SigmaOp}}+\frac{\epsilon}{Q_{\max}^{(\tau)}}\mathbbm{1}\right){\left(\hat{\q}^{\left(\rho\right)}\right)}^{-\frac{1}{2}}\nonumber\\
                                    &=\frac{1}{Q_{\max}^{(\tau)}}\sqrt{\hat{\q}^{\left(\rho\right)}}\kernel\sqrt{\hat{\q}^{\left(\rho\right)}}+\frac{\epsilon}{Q_{\max}^{(\tau)}}\Pi^{\left(\rho\right)},
\end{align}
where we use $\hat{\SigmaOp}=\kernel\hat{\q}^{(\rho)}$, and $\Pi^{\left(\rho\right)}$ is a projector onto the support of $\hat{\q}^{\left(\rho\right)}$.
In case $\hat{\q}^{\left(\rho\right)}$ does not have full rank, ${\left(\hat{\q}^{\left(\rho\right)}\right)}^{-\frac{1}{2}}$ denotes $\sqrt{{\left(\hat{\q}^{\left(\rho\right)}\right)}^{-1}}$, where ${\left(\hat{\q}^{\left(\rho\right)}\right)}^{-1}$ in this case is the Moore-Penrose pseudoinverse of $\hat{\q}^{\left(\rho\right)}$.
We have by definition
\begin{equation}
  {\left(\hat{\SigmaOp}_\epsilon^{\left(\rho\right)}\right)}^{-1}= \sqrt{\hat{\q}^{\left(\rho\right)}}{\left({\frac{1}{Q_{\max}^{(\tau)}}\hat{\SigmaOp}}+\frac{\epsilon}{Q_{\max}^{(\tau)}}\mathbbm{1}\right)}^{-1}{\left(\hat{\q}^{\left(\rho\right)}\right)}^{-\frac{1}{2}}.
\end{equation}
Correspondingly, in the same way as the main text, we let $\hat{\SigmaOp}_\epsilon$ denote a positive definite operator that has the full support on $\mathcal{H}^X$, and coincides with $\hat{\SigmaOp}_\epsilon^{\left(\rho\right)}$ if projected on the support of $\hat{\q}^{\left(\rho\right)}$
\begin{equation}
  \label{seq:hat_Sigma_lambda}
  \hat{\SigmaOp}_\epsilon\coloneqq\frac{1}{Q_{\max}^{(\tau)}}\sqrt{\hat{\q}^{\left(\rho\right)}}\kernel\sqrt{\hat{\q}^{\left(\rho\right)}}+\frac{\epsilon}{Q_{\max}^{(\tau)}}\mathbbm{1}.
\end{equation}
Note that since $\hat{\q}^{\left(\rho\right)}$ is diagonal and $\kernel$ is symmetric,
we have
\begin{equation}
  \label{seq:symmetric_Sigma}
  \hat{\SigmaOp}_\epsilon=\hat{\SigmaOp}_\epsilon^{\mathrm{T}},
\end{equation}
where the right-hand side represents the transpose with respect to the computational basis.
Then, we can rewrite the last line of~\eqref{seq:hat_p_lambda} as
\begin{align}
  \label{seq:p_lambda_ast_simplified}
  Q_{\epsilon}^\ast\left(\frac{\tilde{x}}{G}\right)P^{\left(\tau\right)}\left(\frac{\tilde{x}}{G}\right)
  &\propto{\Braket{\tilde{x}|\sqrt{\frac{1}{Q_{\max}^{(\tau)}}\Q^{(\tau)}}\F_{D}^\dag \hat{\q}^{\left(\rho\right)}{\left({\frac{1}{Q_{\max}^{(\tau)}}\hat{\SigmaOp}}+\frac{\epsilon}{Q_{\max}^{(\tau)}}\mathbbm{1}\right)}^{-1} \F_{D} \sqrt{\frac{1}{Q_{\max}^{(\tau)}}\Q^{(\tau)}}| \tilde{x}}}\nonumber\\
  &={\Braket{\tilde{x}|\sqrt{\frac{1}{Q_{\max}^{(\tau)}}\Q^{(\tau)}}\F_{D}^\dag \sqrt{\hat{\q}^{\left(\rho\right)}}{\left(\hat{\SigmaOp}_\epsilon^{\left(\rho\right)}\right)}^{-1} \sqrt{\hat{\q}^{\left(\rho\right)}} \F_{D} \sqrt{\frac{1}{Q_{\max}^{(\tau)}}\Q^{(\tau)}}| \tilde{x}}}\nonumber\\
  &={\Braket{\tilde{x}|\sqrt{\frac{1}{Q_{\max}^{(\tau)}}\Q^{(\tau)}}\F_{D}^\dag \sqrt{\hat{\q}^{\left(\rho\right)}}\hat{\SigmaOp}_\epsilon^{-1} \sqrt{\hat{\q}^{\left(\rho\right)}} \F_{D} \sqrt{\frac{1}{Q_{\max}^{(\tau)}}\Q^{(\tau)}}| \tilde{x}}},
\end{align}
where this probability distribution is normalized by $\sum_{\tilde{x}\in\tilde{\mathcal{X}}}Q_{\epsilon}^\ast\left(\frac{\tilde{x}}{G}\right)P^{\left(\tau\right)}\left(\frac{\tilde{x}}{G}\right)=1$.

To prove the proposition, we analyze the probability distribution obtained from the measurement of the quantum state
\begin{equation}
  \label{seq:state}
  \Ket{\Psi}^{XX^\prime}\propto\sum_{\tilde{x}\in\tilde{\mathcal{X}}}\hat{\SigmaOp}_\epsilon^{-\frac{1}{2}}\Ket{\tilde{x}}^X\otimes \sqrt{\frac{1}{Q_{\max}^{(\tau)}}\Q^{(\tau)}}\F_{D}^\dag\sqrt{\hat{\q}^{\left(\rho\right)}}\Ket{\tilde{x}}^{X^\prime}\in\mathcal{H}^X\otimes\mathcal{H}^{X^\prime}.
\end{equation}
where $\sqrt{\hat{\q}^{\left(\rho\right)}}\Ket{\tilde{x}}^{X^\prime}=\sqrt{\hat{q}^{\left(\rho\right)}\left(\tilde{x}\right)}\Ket{\tilde{x}}^{X^\prime}$.
For any operators $\A$ on $\mathcal{H}^X$ and $\B$ on $\mathcal{H}^{X^\prime}$ where the dimensions of these Hilbert spaces are the same
\begin{equation}
  \dim\mathcal{H}^X=\dim\mathcal{H}^{X^\prime},
\end{equation}
a straightforward linear algebraic calculation shows~\cite{N4}
\begin{equation}
  \sum_{\tilde{x}\in\tilde{\mathcal{X}}}\A\Ket{\tilde{x}}^X\otimes \B\Ket{\tilde{x}}^{X^\prime}=\sum_{\tilde{x}\in\tilde{\mathcal{X}}}\Ket{\tilde{x}}^X\otimes \B\A^\mathrm{T}\Ket{\tilde{x}}^{X^\prime}.
\end{equation}
Applying this equality to~\eqref{seq:state}, we have
\begin{align}
  \Ket{\Psi}^{XX^\prime}&\propto\sum_{\tilde{x}\in\tilde{\mathcal{X}}}\hat{\SigmaOp}_\epsilon^{-\frac{1}{2}}\Ket{\tilde{x}}^X\otimes \sqrt{\frac{1}{Q_{\max}^{(\tau)}}\Q^{(\tau)}}\F_{D}^\dag\sqrt{\hat{\q}^{\left(\rho\right)}}\Ket{\tilde{x}}^{X^\prime}\nonumber\\
  &=\sum_{\tilde{x}\in\tilde{\mathcal{X}}}\Ket{\tilde{x}}^X\otimes \sqrt{\frac{1}{Q_{\max}^{(\tau)}}\Q^{(\tau)}}\F_{D}^\dag\sqrt{\hat{\q}^{\left(\rho\right)}}{\left(\hat{\SigmaOp}_\epsilon^\mathrm{T}\right)}^{-\frac{1}{2}}\Ket{\tilde{x}}^{X^\prime}\nonumber\\
  &=\sum_{\tilde{x}\in\tilde{\mathcal{X}}}\Ket{\tilde{x}}^X\otimes \sqrt{\frac{1}{Q_{\max}^{(\tau)}}\Q^{(\tau)}}\F_{D}^\dag\sqrt{\hat{\q}^{\left(\rho\right)}}\hat{\SigmaOp}_\epsilon^{-\frac{1}{2}}\Ket{\tilde{x}}^{X^\prime},
\end{align}
where the last line follows from~\eqref{seq:symmetric_Sigma}.
If we performed a measurement of $\Ket{\Psi}^{XX^\prime}$ in the computational basis $\left\{\Ket{\tilde{x}}^X\otimes\Ket{\tilde{x}^\prime}^{X^\prime}\right\}$, the probability distribution of measurement outcomes, i.e., the square of the amplitude as summarized in Sec.~\ref{sec:quantum_suplementary}, would be
\begin{align}
  p(\tilde{x},\tilde{x}^\prime)&=\left|\left(\Bra{\tilde{x}}^X\otimes\Bra{\tilde{x}^\prime}^{X^\prime}\right)\Ket{\Psi}^{XX^\prime}\right|^2\nonumber\\
               &\propto\left|\Braket{\tilde{x}^\prime|\sqrt{\frac{1}{Q_{\max}^{(\tau)}}\Q^{(\tau)}}\F_{D}^\dag\sqrt{\hat{\q}^{\left(\rho\right)}}\hat{\SigmaOp}_\epsilon^{-\frac{1}{2}}|{\tilde{x}}}\right|^2\nonumber\\
               &=\Braket{\tilde{x}^\prime|\sqrt{\frac{1}{Q_{\max}^{(\tau)}}\Q^{(\tau)}}\F_{D}^\dag\sqrt{\hat{\q}^{\left(\rho\right)}}\hat{\SigmaOp}_\epsilon^{-\frac{1}{2}}|\tilde{x}}\Braket{\tilde{x}|\hat{\SigmaOp}_\epsilon^{-\frac{1}{2}}\sqrt{\hat{\q}^{\left(\rho\right)}}\F_{D}\sqrt{\frac{1}{Q_{\max}^{(\tau)}}\Q^{(\tau)}}|\tilde{x}^\prime}.
\end{align}
Since a measurement of the register $X^\prime$ yields an outcome $\tilde{x}^\prime$ with probability $p\left(\tilde{x}^\prime\right)=\sum_{\tilde{x}\in\tilde{\mathcal{X}}}p\left(\tilde{x},\tilde{x}^\prime\right)$ as summarized in Sec.~\ref{sec:quantum_suplementary},
we obtain
\begin{align}
  \label{seq:p_tilde_x}
  p\left(\tilde{x}^\prime\right)&\propto\sum_{\tilde{x}\in\tilde{\mathcal{X}}}\Braket{\tilde{x}^\prime|\sqrt{\frac{1}{Q_{\max}^{(\tau)}}\Q^{(\tau)}}\F_{D}^\dag\sqrt{\hat{\q}^{\left(\rho\right)}}\hat{\SigmaOp}_\epsilon^{-\frac{1}{2}}|\tilde{x}}\Braket{\tilde{x}|\hat{\SigmaOp}_\epsilon^{-\frac{1}{2}}\sqrt{\hat{\q}^{\left(\rho\right)}}\F_{D}\sqrt{\frac{1}{Q_{\max}^{(\tau)}}\Q^{(\tau)}}|\tilde{x}^\prime}\nonumber\\
                                &=\Braket{\tilde{x}^\prime|\sqrt{\frac{1}{Q_{\max}^{(\tau)}}\Q^{(\tau)}}\F_{D}^\dag\sqrt{\hat{\q}^{\left(\rho\right)}}\hat{\SigmaOp}_\epsilon^{-\frac{1}{2}}\left(\sum_{\tilde{x}\in\tilde{\mathcal{X}}}\Ket{\tilde{x}}\Bra{\tilde{x}}\right)\hat{\SigmaOp}_\epsilon^{-\frac{1}{2}}\sqrt{\hat{\q}^{\left(\rho\right)}}\F_{D}\sqrt{\frac{1}{Q_{\max}^{(\tau)}}\Q^{(\tau)}}|\tilde{x}^\prime}\nonumber\\
                                &=\Braket{\tilde{x}^\prime|\sqrt{\frac{1}{Q_{\max}^{(\tau)}}\Q^{(\tau)}}\F_{D}^\dag\sqrt{\hat{\q}^{\left(\rho\right)}}\hat{\SigmaOp}_\epsilon^{-\frac{1}{2}}\mathbbm{1}\hat{\SigmaOp}_\epsilon^{-\frac{1}{2}}\sqrt{\hat{\q}^{\left(\rho\right)}}\F_{D}\sqrt{\frac{1}{Q_{\max}^{(\tau)}}\Q^{(\tau)}}|\tilde{x}^\prime}\nonumber\\
                                &=\Braket{\tilde{x}^\prime|\sqrt{\frac{1}{Q_{\max}^{(\tau)}}\Q^{(\tau)}}\F_{D}^\dag\sqrt{\hat{\q}^{\left(\rho\right)}}\hat{\SigmaOp}_\epsilon^{-1}\sqrt{\hat{\q}^{\left(\rho\right)}}\F_{D}\sqrt{\frac{1}{Q_{\max}^{(\tau)}}\Q^{(\tau)}}|\tilde{x}^\prime}.
\end{align}
Recall that the normalization $\left\|\Ket{\Psi}^{XX^\prime}\right\|_2=1$ of the quantum state yields $\sum_{\tilde{x}^\prime\in\tilde{\mathcal{X}}}p\left(\tilde{x}^\prime\right)=1$.
Therefore,~\eqref{seq:p_lambda_ast_simplified} and~\eqref{seq:p_tilde_x} yield
\begin{equation}
  p\left(\tilde{x}\right)=
  Q_{\epsilon}^\ast\left(\frac{\tilde{x}}{G}\right)P^{\left(\tau\right)}\left(\frac{\tilde{x}}{G}\right),
\end{equation}
which shows the conclusion.
\end{proof}

\section{\label{sec:3}Quantum algorithm for sampling an optimized random feature}

In this section, we show our quantum algorithm for sampling an optimized random feature and bound its runtime.

Algorithm~\ref{salg:random_feature_revised} shows our quantum algorithm.
Note that each line of Algorithm~\ref{salg:random_feature_revised} is performed approximately with a sufficiently small precision to achieve the overall sampling precision $\Delta>0$, in the same way as classical algorithms that deal with real number using fixed- or floating-point number representation with a sufficiently small precision.
In Algorithm~\ref{salg:random_feature_revised},
we represent computation of the function $Q^{(\tau)}$ as a quantum oracle $\mathcal{O}_\tau$.
This oracle $\mathcal{O}_\tau$ computes $Q^{(\tau)}$ while maintaining the superpositions (i.e., linear combinations) in a given quantum state, that is,
\begin{equation}
  \label{seq:oracle_tau_definition}
  \mathcal{O}_\tau\Big(\sum_{v}\alpha_v\Ket{v}\otimes\Ket{0}\Big)=\sum_{v}\alpha_v\ket{v}\otimes\ket{Q^{(\tau)}\left(v\right)},
\end{equation}
where $\alpha_v\in\mathbb{C}$ can be any coefficient of the given state, and $\ket{v}$ and $\ket{Q^{(\tau)}\left(v\right)}$ are computational-basis states corresponding to bit strings representing $v\in\mathcal{V}$ and $Q^{(\tau)}\left(v\right)\in\mathbb{R}$ in the fixed-point number representation with sufficient precision.
We use $\mathcal{O}_\tau$ for simplicity of the presentation, and unlike $\mathcal{O}_\rho$ given by~\eqref{seq:oracle_rho_definition}, $\mathcal{O}_\tau$ is not a black box in our quantum algorithm since we can implement $\mathcal{O}_\tau$ explicitly by a quantum circuit under the assumption in the main text, without using QRAM discussed in Sec.~\ref{sec:oracles}.
Using the assumption that classical computation can evaluate the function $Q^{(\tau)}$ efficiently in a short time denoted by
\begin{equation}
  \label{seq:requirement_T_tau}
T_\tau=O(\poly(D)),
\end{equation}
we can efficiently implement $\mathcal{O}_\tau$ in runtime $O(T_\tau)$;
in particular, if we can compute $Q^{(\tau)}$ by numerical libraries using arithmetics in runtime $T_\tau$,
then a quantum computer can also perform the same arithmetics to implement $\mathcal{O}_\tau$ in runtime $O(T_\tau)$~\cite{haner2018}.
Note that even if the numerical libraries evaluated $Q^{(\tau)}$ by means of a lookup table stored in RAM, quantum computers could instead use the QRAM\@.
In the following, for simplicity of the presentation, the runtime of $\mathcal{O}_\tau$ per query may also be denoted by $T_\tau$ (with abuse of notation) as our runtime analysis ignores constant factors.
Note that in the case of the Gaussian kernel and the Laplacian kernel, $Q^{(\tau)}$ is given in terms of special functions as shown in the main text,
and we have
\begin{equation}
  T_\tau=O(D),
\end{equation}
which satisfies~\eqref{seq:requirement_T_tau}.

In the rest of this section, we prove the following theorem that bounds the runtime of Algorithm~\ref{salg:random_feature_revised}.

\begin{theorem}[\label{sthm:sampling}Runtime of our quantum algorithm for sampling an optimized random feature]
  Given $D$-dimensional data discretized by $G>0$,
  for any learning accuracy ${\epsilon}>0$
  and any sampling precision $\Delta>0$,
  the runtime $T_1$ of Algorithm~\ref{salg:random_feature_revised} for sampling each optimized random feature
   $v_{{G}}\in\mathcal{V}_{{G}}$ from a distribution $Q(v_G)P^{(\tau)}(v_G)$ close to the optimized distribution $Q_\epsilon^\ast(v_G)P^{(\tau)}(v_G)$ with precision
  $\sum_{v_{{G}}\in\mathcal{V}_{{G}}}{|Q(v_G)P^{(\tau)}(v_G)-Q_\epsilon^\ast(v_G)P^{(\tau)}(v_G)|}\leqq\Delta$ is
  \begin{align*}
    T_1=O\left(D\log\left({G}\right)\log\log\left({G}\right)+T_\rho+T_\tau\right)\times\widetilde{O}\left(\frac{Q_{\max}^{\left(\tau\right)}}{\epsilon}\polylog\left(\frac{1}{\Delta}\right)\right),
  \end{align*}
  where $T_\rho$ and $T_\tau$ are the runtime of the quantum oracles $\mathcal{O}_\rho$ and $\mathcal{O}_\tau$ per query, and $Q_{\max}^{(\tau)}$, $\mathcal{O}_\rho$ and $\mathcal{O}_\tau$ are defined as~\eqref{seq:q_tau_g_max},~\eqref{seq:oracle_rho_definition}, and~\eqref{seq:oracle_tau_definition}, respectively.
\end{theorem}

\begin{algorithm}[H]
  \caption{\label{salg:random_feature_revised}Quantum algorithm for sampling an optimized random feature (quOptRF).}
  \begin{algorithmic}[1]
    \REQUIRE{A desired accuracy ${\epsilon}>0$ in the supervised learning, sampling precision $\Delta>0$, quantum oracles $\mathcal{O}_\rho$ in~\eqref{seq:oracle_rho_definition} and $\mathcal{O}_\tau$ in~\eqref{seq:oracle_tau_definition}, and $Q_{\max}^{\left(\tau\right)}>0$ in~\eqref{seq:q_tau_g_max}.}
    \ENSURE{An optimized random feature $v_{{G}}\in\mathcal{V}_{{G}}$ sampled from a probability distribution $Q\left(v_{{G}}\right)P^{\left(\tau\right)}\left(v_G\right)$ with ${\sum_{v_{{G}}\in\mathcal{V}_G}|Q\left(v_{{G}}\right)P^{\left(\tau\right)}\left(v_G\right)-Q_\epsilon^\ast\left(v_{{G}}\right)P^{\left(\tau\right)}\left(v_G\right)|}\leqq\Delta$.}
    \STATE{Initialize quantum registers $X$ and $X'$, load data onto  $X'$ by $\mathcal{O}_\rho$, and perform \textsc{CNOT} gates on $X$ and $X'$
      \begin{equation}
        \Ket{0}^X\otimes\Ket{0}^{X^\prime}\xrightarrow{\mathcal{O}_\rho}\sum_{\tilde{x}\in\tilde{\mathcal{X}}}\Ket{0}^X\otimes\sqrt{\hat{\q}^{\left(\rho\right)}}\Ket{\tilde{x}}^{X^\prime}\xrightarrow{\textsc{CNOT}}\sum_{\tilde{x}\in\tilde{\mathcal{X}}}\Ket{\tilde{x}}^X\otimes\sqrt{\hat{\q}^{\left(\rho\right)}}\Ket{\tilde{x}}^{X^\prime}.
      \end{equation}
    }
    \STATE{Perform $\F_D^\dag$ on $X^\prime$ by QFT~\cite{C5} to obtain
      \begin{equation}
      \sum_{\tilde{x}\in\tilde{\mathcal{X}}}\Ket{\tilde{x}}^X\otimes \F_D^\dag\sqrt{\hat{\q}^{\left(\rho\right)}}\Ket{\tilde{x}}^{X^\prime}.
      \end{equation}
    }
    \STATE{Apply the block encoding of $\sqrt{\frac{1}{Q_{\max}^{\left(\tau\right)}}\Q^{\left(\tau\right)}}$ (Lemma~\ref{sprp:block_encoding_q_tau}) to $X^\prime$ followed by amplitude amplification to obtain a state proportional to 
      \begin{equation}
        \sum_{\tilde{x}\in\tilde{\mathcal{X}}}\Ket{\tilde{x}}^X\otimes \sqrt{\frac{1}{Q_{\max}^{\left(\tau\right)}}\Q^{\left(\tau\right)}}\F_D^\dag\sqrt{\hat{\q}^{\left(\rho\right)}}\Ket{\tilde{x}}^{X^\prime}.
      \end{equation}
    }
    \STATE{Apply the block encoding of $\hat{\SigmaOp}_{\epsilon}^{-\frac{1}{2}}$ (Lemma~\ref{sprp:block_encoding_sigma_lambda}) to $X$ to obtain the quantum state $\Ket{\Psi}^{XX^\prime}$ in Proposition~\ref{sprp:state}.
    }
    \COMMENT{This step requires our technical contribution since \textit{no assumption on sparsity and low rank} is imposed on $\hat{\SigmaOp}_{\epsilon}$.}
    \STATE{Perform a measurement of $X'$ in the computational basis to obtain $\tilde{x}$ with probability $Q_\epsilon^\ast\left(\frac{\tilde{x}}{G}\right)P^{\left(\tau\right)}\left(\frac{\tilde{x}}{G}\right)$.}
    \STATE{\textbf{Return }$v_{{G}}=\frac{\tilde{x}}{{G}}$.}
\end{algorithmic}
\end{algorithm}

To prove Theorem~\ref{sthm:sampling}, in the following, we construct efficient implementations of block encodings; in particular, we first show a block encoding of $\sqrt{\frac{1}{Q_{\max}^{(\tau)}}\Q^{(\tau)}}$,
and then using this block encoding, we show that of $\hat{\SigmaOp}_\epsilon$.
Then, we will provide the runtime analysis of Algorithm~\ref{salg:random_feature_revised} using these block encodings.
In Algorithm~\ref{salg:random_feature_revised}, we combine these block encodings with two fundamental subroutines of quantum algorithms, namely, quantum Fourier transform (QFT)~\cite{C5,H2} and quantum singular value transformation (QSVT)~\cite{G1}.
Using QFT, we can implement the unitary operator $\F$ defined as~\eqref{seq:F} with precision $\Delta$ by a quantum circuit composed of $O\left(\log\left({G}\right)\log\left(\frac{\log G}{\Delta}\right)\right)$ gates~\cite{C5}.
Thus, we can implement $\F_{D}=F^{\otimes D}$ defined as~\eqref{seq:F_D} by a quantum circuit composed of gates of order
\begin{equation}
  \label{seq:runtime_F_D}
  O\left(D\log\left({G}\right)\log\left(\frac{\log G}{\Delta}\right)\right).
\end{equation}
Note that QFT in Ref.~\cite{C5} that we use in the following analysis has slightly better runtime than QFT in Ref.~\cite{H2} by a poly-logarithmic term in $\frac{1}{\Delta}$, but we may also use QFT in Ref.~\cite{H2} without changing any statement of our lemmas and theorems since this poly-logarithmic term is not dominant.
In our analysis, we multiply two numbers represented by $O\left(\log\left({G}\right)\right)$ bits using the algorithm shown in Ref.~\cite{harvey:hal-02070778} within time
\begin{equation}
  \label{seq:runtime_multiplication}
  O\left(\log\left({G}\right)\log\log\left({G}\right)\right),
\end{equation}
which we can perform also on quantum computer by implementing arithmetics using a quantum circuit~\cite{haner2018}.
Note that we could also use exact quantum Fourier transform~\cite{N4} or grammar-school-method multiplication instead of these algorithms in Refs.~\cite{C5,harvey:hal-02070778}, to decrease a constant factor in the runtime of our algorithm at the expense of logarithmically increasing the asymptotic scaling in terms of ${G}$ from $\log\left({G}\right)\log\log\left({G}\right)$ to ${\left(\log\left({G}\right)\right)}^2$.
In the following runtime analysis,
we use the definition~\eqref{seq:block_encoding_definition} of block encoding to clarify the dependency on precision $\Delta$,
and the quantum registers for storing real number use fixed-point number representation with sufficient precision $O(\Delta)$ to achieve the overall precision $\Delta$.

Our construction of block encodings of $\sqrt{\frac{1}{Q_{\max}^{(\tau)}}\Q^{(\tau)}}$ and $\hat{\SigmaOp}_\epsilon$ is based on a prescription of constructing a block encoding from a quantum circuit for implementing a measurement described by a positive operator-valued measure (POVM)~\cite{G1}.
In particular, for any precision $\Delta>0$ and any POVM operator $\Lam$, that is, an operator satisfying $0\leqq \Lam\leqq\mathbbm{1}$, let $\U$ be a unitary operator represented by a quantum circuit that satisfies for any state $\Ket{\psi}$
\begin{equation}
  \label{seq:povm_implementation}
  \left|\tr\left[\Ket{\psi}\Bra{\psi}\Lam\right]-\tr\left[\U\left(\Ket{0}\Bra{0}^{\otimes n}\otimes\Ket{\psi}\Bra{\psi}\right)\U^\dag\left(\Ket{0}\Bra{0}^{\otimes 1}\otimes\mathbbm{1}\right)\right]\right|\leqq \Delta,
\end{equation}
where $\Ket{0}^{\otimes n}$ is a fixed state of $n$ auxiliary qubits.
The quantum circuit $\U$ in~\eqref{seq:povm_implementation} means that $\U$ implements a quantum measurement represented by the POVM operator $\Lam$ with precision $\Delta$;
that is,
given any input state $\Ket{\psi}$ and $n$ auxiliary qubits initially prepared in $\Ket{0}^{\otimes n}$,
if we perform the circuit $\U$ to obtain a state $\U\left(\Ket{0}^{\otimes n}\otimes\Ket{\psi}\right)$ and perform a measurement of one of the qubits for the obtained state in the computational basis $\{\Ket{0},\Ket{1}\}$,
then we obtain a measurement outcome $0$ with probability
\begin{equation}
  \tr\left[\U\left(\Ket{0}\Bra{0}^{\otimes n}\otimes\Ket{\psi}\Bra{\psi}\right)\U^\dag\left(\Ket{0}\Bra{0}^{\otimes 1}\otimes\mathbbm{1}\right)\right].
\end{equation}
Then, it is known that we can construct a $\left(1,1+n,\Delta\right)$-block encoding of $\Lam$ using one $\U$, one $\U^\dag$, and one quantum logic gate (i.e., the controlled $\textsc{NOT}$ ($\textsc{CNOT}$) gate)~\cite{G1}.
The $\textsc{CNOT}$ gate is defined as a two-qubit unitary operator
\begin{equation}
  \textsc{CNOT}\coloneqq\Ket{0}\Bra{0}\otimes\mathbbm{1}+\Ket{1}\Bra{1}\otimes\sigma_x,
\end{equation}
where the first qubit is a controlled qubit, the second qubit is a target qubit, and $\sigma_x$ is a Pauli unitary operator
\begin{equation}
  \label{seq:sigma_x}
  \sigma_x\coloneqq\Ket{0}\Bra{1}+\Ket{1}\Bra{0}.
\end{equation}
The $\textsc{CNOT}$ gate acts as
\begin{equation}
  \textsc{CNOT}\left(\left(\alpha_0\Ket{0}+\alpha_1\Ket{1}\right)\otimes\Ket{0}\right)=\alpha_0\Ket{0}\otimes\Ket{0}+\alpha_1\Ket{1}\otimes\Ket{1}.
\end{equation}

For a given POVM operator $\Lam$, no general way of constructing the circuit representing $\U$ in~\eqref{seq:povm_implementation} has been shown in Ref.~\cite{G1};
in contrast, we here explicitly construct the circuit for a diagonal POVM operator $\Lam=\sqrt{\frac{1}{Q_{\max}^{(\tau)}}\Q^{(\tau)}}$ in the following lemma, using the quantum oracle $\mathcal{O}_\tau$.
Note that since a diagonal operator is sparse, a conventional way of implementing the block encoding of a sparse operator~\cite{G1} would also be applicable to construct a block encoding of $\sqrt{\frac{1}{Q_{\max}^{(\tau)}}\Q^{(\tau)}}$; however, our key contribution here is to use the circuit for the block encoding of $\sqrt{\frac{1}{Q_{\max}^{(\tau)}}\Q^{(\tau)}}$ as a building block of a more complicated block encoding, i.e., the block encoding of $\hat{\SigmaOp}_\epsilon$, which is not necessarily sparse or of low rank.

\begin{lemma}[\label{sprp:block_encoding_q_tau}Block encoding of a diagonal POVM operator]
  For any diagonal positive semidefinite operator $\Q^{(\tau)}$ defined as~\eqref{seq:q_tau_g},
  we can implement a $\left(1,O\left(D\log\left({G}\right)\polylog\left(\frac{1}{\Delta}\right)\right),\Delta\right)$-block encoding of $\sqrt{\frac{1}{Q_{\max}^{(\tau)}}\Q^{(\tau)}}$ by a quantum circuit composed of $O\left(D\log\left({G}\right)\log\log\left({G}\right)\polylog\left(\frac{1}{\Delta}\right)\right)$ gates and one query to the quantum oracle $\mathcal{O}_\tau$.
\end{lemma}

\begin{figure}[t]
  \centering
  \includegraphics[width=5.5in]{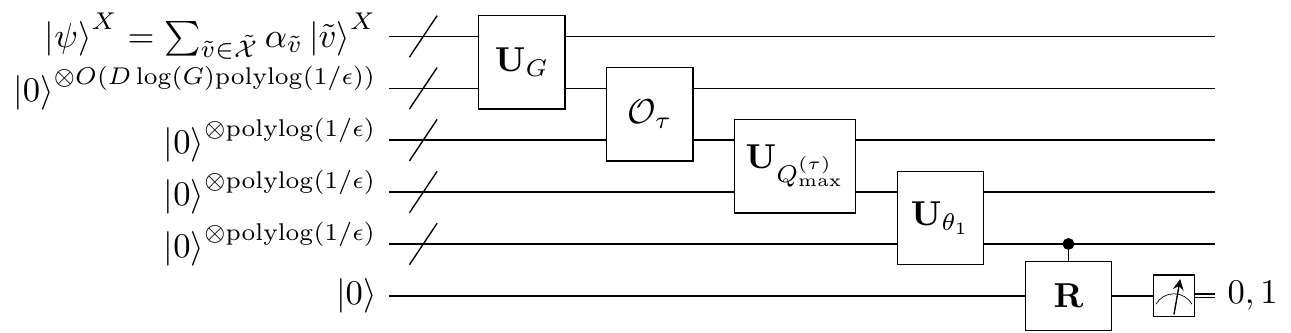}
  \caption{\label{fig:q_tau}A quantum circuit representing a unitary operator $\U$ that achieves~\eqref{seq:povm_implementation} for $\Lam=\sqrt{\frac{1}{Q_{\max}^{(\tau)}}\Q^{(\tau)}}$, which can be used for implementing a block encoding of $\sqrt{\frac{1}{Q_{\max}^{(\tau)}}\Q^{(\tau)}}$. This circuit achieves the transformation of quantum states shown in a chain starting from~\eqref{seq:0}. The last controlled gate represents $\textsc{C}\R$. Regarding the notations on quantum circuits, see, e.g.,~\cite{N4}.}
\end{figure}

\begin{proof}
  We construct a quantum circuit representing a unitary operator $\U$ that achieves~\eqref{seq:povm_implementation} for $\Lam=\sqrt{\frac{1}{Q_{\max}^{(\tau)}}\Q^{(\tau)}}$.
  We write the input quantum state as
  \begin{equation}
    \Ket{\psi}^X=\sum_{\tilde{v}\in\tilde{\mathcal{X}}}\alpha_{\tilde{v}}\Ket{\tilde{v}}^X\in\mathcal{H}^X.
  \end{equation}
  Define a function
  \begin{equation}
    \theta_1\left(q\right)\coloneqq\arccos\left({q}^{\frac{1}{4}}\right).
  \end{equation}
  Define unitary operators $\U_{G}$, $\U_{Q_{\max}^{(\tau)}}$, and $\U_{\theta_1}$ acting as
  \begin{align}
    \U_{G}:&\Ket{x}\otimes\Ket{0}\xrightarrow{\U_{G}}\Ket{x}\otimes\Ket{\frac{x}{G}},\\
    \U_{Q_{\max}^{(\tau)}}:&\Ket{x}\otimes\Ket{0}\xrightarrow{\U_{Q_{\max}^{(\tau)}}}\Ket{x}\otimes\Ket{\frac{x}{Q_{\max}^{(\tau)}}},\\
    \U_{\theta_1}:&\Ket{x}\otimes\Ket{0}\xrightarrow{\U_{\theta_1}}\Ket{x}\otimes\Ket{\theta_1(x)}.
  \end{align}
  Let $\R_\theta$ denote a unitary operator representing a one-qubit rotation
  \begin{equation}
    \label{seq:rotation}
    \R_\theta\coloneqq\left(\begin{matrix}\cos\theta & -\sin\theta \\
    \sin\theta & \cos\theta\end{matrix}\right),
  \end{equation}
  and a controlled rotation $\textsc{C}\R$ is defined as
  \begin{equation}
    \label{seq:controlled_rotation_q_tau}
    \textsc{C}\R=\sum_{\theta}\Ket{\theta}\Bra{\theta}\otimes \R_\theta.
  \end{equation}

  Using these notations, we show a quantum circuit representing $\U$ in Fig.~\ref{fig:q_tau}.
  This circuit achieves the following transformation up to precision $\Delta$
  \begin{align}
    \label{seq:0}
    &\Ket{\psi}\otimes\Ket{0}\otimes\Ket{0}\otimes\Ket{0}\otimes\Ket{0}\otimes\Ket{0}\\
    \label{seq:1}
    \xrightarrow{\U_{G}}&\sum_{\tilde{v}\in\tilde{\mathcal{X}}}\alpha_{\tilde{v}}\Ket{\tilde{v}}\otimes\Ket{\frac{\tilde{v}}{G}}\otimes\Ket{0}\otimes\Ket{0}\otimes\Ket{0}\otimes\Ket{0}\\
    \label{seq:oracle_tau_used}
    \xrightarrow{\mathcal{O}_\tau}&\sum_{\tilde{v}\in\tilde{\mathcal{X}}}\alpha_{\tilde{v}}\Ket{\tilde{v}}\otimes\Ket{\frac{\tilde{v}}{G}}\otimes\Ket{Q^{(\tau)}\left(\frac{\tilde{v}}{G}\right)}\otimes\Ket{0}\otimes\Ket{0}\otimes\Ket{0}\\
    \label{seq:2}
    \xrightarrow{\U_{Q_{\max}^{(\tau)}}}&\sum_{\tilde{v}\in\tilde{\mathcal{X}}}\alpha_{\tilde{v}}\Ket{\tilde{v}}\otimes\Ket{\frac{\tilde{v}}{G}}\otimes\Ket{Q^{(\tau)}\left(\frac{\tilde{v}}{G}\right)}\otimes\Ket{\frac{Q^{(\tau)}\left(\frac{\tilde{v}}{G}\right)}{Q_{\max}^{(\tau)}}}\otimes\Ket{0}\otimes\Ket{0}\\
    \label{seq:3}
    \xrightarrow{\U_{\theta_1}}&\sum_{\tilde{v}\in\tilde{\mathcal{X}}}\alpha_{\tilde{v}}\Ket{\tilde{v}}\otimes\Ket{\frac{\tilde{v}}{G}}\otimes\Ket{Q^{(\tau)}\left(\frac{\tilde{v}}{G}\right)}\otimes\Ket{\frac{Q^{(\tau)}\left(\frac{\tilde{v}}{G}\right)}{Q_{\max}^{(\tau)}}}\otimes\Ket{\theta_1\left(\frac{Q^{(\tau)}\left(\frac{\tilde{v}}{G}\right)}{Q_{\max}^{(\tau)}}\right)}\otimes\Ket{0}\\
    \label{seq:4}
    \xrightarrow{\textsc{C}\R}&\sum_{\tilde{v}\in\tilde{\mathcal{X}}}\alpha_{\tilde{v}}\Ket{\tilde{v}}\otimes\Ket{\frac{\tilde{v}}{G}}\otimes\Ket{Q^{(\tau)}\left(\frac{\tilde{v}}{G}\right)}\otimes\Ket{\frac{Q^{(\tau)}\left(\frac{\tilde{v}}{G}\right)}{Q_{\max}^{(\tau)}}}\otimes\Ket{\theta_1\left(\frac{Q^{(\tau)}\left(\frac{\tilde{v}}{G}\right)}{Q_{\max}^{(\tau)}}\right)}\otimes\nonumber\\
                                    &\quad\left({\left(\frac{1}{Q_{\max}^{(\tau)}}Q^{(\tau)}\left(\frac{\tilde{v}}{G}\right)\right)}^{\frac{1}{4}}\Ket{0}+\sqrt{1-\sqrt{\frac{1}{Q_{\max}^{(\tau)}}Q^{(\tau)}\left(\frac{\tilde{v}}{G}\right)}}\Ket{1}\right),
  \end{align}
  where the quantum registers for storing real number use fixed-point number representation with sufficient precision $O(\Delta)$ to achieve the overall precision $\Delta$ in~\eqref{seq:povm_implementation}.
  In~\eqref{seq:1}, each of the $D$ elements of the vector $\tilde{v}$ in the first quantum register is multiplied by $\frac{1}{G}$ using arithmetics, and the result is stored in the second quantum register.
  Since $\frac{1}{G}$ can be approximately represented with precision $O(\Delta)$ using $O\left(\log\left(\frac{1}{\Delta}\right)\right)$ bits, these $D$ multiplications take $O\left(D\log\left({G}\right)\log\log\left({G}\right)\polylog\left(\frac{1}{\Delta}\right)\right)$ time
  due to~\eqref{seq:runtime_multiplication}, which is dominant.
  The runtime of the quantum oracle $\mathcal{O}_\tau$ queried in~\eqref{seq:oracle_tau_used} is $T_\tau$.
  We can multiply $\frac{1}{Q_{\max}^{(\tau)}}$ in~\eqref{seq:2} and calculate the elementary function $\theta_1$ in~\eqref{seq:3} up to precision $O(\Delta)$ by arithmetics within time $O\left(\polylog\left(\frac{1}{\Delta}\right)\right)$~\cite{haner2018}.
  In~\eqref{seq:4}, we apply $\textsc{C}\R$ defined as~\eqref{seq:controlled_rotation_q_tau} to the last qubit controlled by the second last quantum register,
  which uses $O\left(\polylog\left(\frac{1}{\Delta}\right)\right)$ gates since $\Ket{\theta}$ stored in the second last register consists of $O\left(\polylog\left(\frac{1}{\Delta}\right)\right)$ qubits.
  The measurement of the last qubit of~\eqref{seq:4} in the computational basis $\{\Ket{0},\Ket{1}\}$ yields the outcome $0$ with probability
  \begin{equation}
    \label{seq:probability_sqrt_q_tau}
    \sum_{\tilde{v}\in\tilde{\mathcal{X}}}{\left|\alpha_{\tilde{v}}\right|}^2 \sqrt{\frac{1}{Q_{\max}^{(\tau)}}Q^{(\tau)}\left(\frac{\tilde{v}}{G}\right)}=\tr\left[\Ket{\psi}\Bra{\psi}\sqrt{\frac{1}{Q_{\max}^{(\tau)}}\Q^{(\tau)}}\right],
  \end{equation}
  which achieves~\eqref{seq:povm_implementation} for $\Lam=\sqrt{\frac{1}{Q_{\max}^{(\tau)}}\Q^{(\tau)}}$ within the claimed runtime.
\end{proof}

Using the block encoding of $\sqrt{\frac{1}{Q_{\max}^{(\tau)}}\Q^{(\tau)}}$ as a building block, we construct a block encoding of $\hat{\SigmaOp}_\epsilon$ in the following.
Note that while the following proposition provides a $\left(1,O\left(D\log\left({G}\right)\polylog\left(\frac{1}{\Delta}\right)\right),\Delta\right)$-block encoding of $\frac{1}{1+\left(\epsilon/Q_{\max}^{(\tau)}\right)}\hat{\SigmaOp}_\epsilon$, this block encoding is equivalently a $\left(1+\left(\epsilon/Q_{\max}^{(\tau)}\right),O\left(D\log\left({G}\right)\polylog\left(\frac{1}{\Delta}\right)\right),\left(1+\left(\epsilon/Q_{\max}^{(\tau)}\right)\right)\Delta\right)$-block encoding of $\hat{\SigmaOp}_\epsilon$ by definition.
In implementing the block encoding of $\hat{\SigmaOp}_\epsilon$, we use the quantum oracle $\mathcal{O}_\rho$ defined as~\eqref{seq:oracle_rho_definition} in addition to $\mathcal{O}_\tau$.

\begin{lemma}[\label{sprp:block_encoding_sigma_lambda}Block encoding of $\hat{\SigmaOp}_\epsilon$]
  For any $\epsilon>0$
  and any operator $\hat{\SigmaOp}_\epsilon$ given in the form of~\eqref{seq:hat_Sigma_lambda},
  we can implement a $\left(1,O\left(D\log\left({G}\right)\polylog\left(\frac{1}{\Delta}\right)\right),\Delta\right)$-block encoding of
  \begin{equation*}
    \frac{1}{1+\left(\epsilon/Q_{\max}^{(\tau)}\right)}\hat{\SigmaOp}_\epsilon
  \end{equation*}
  by a quantum circuit composed of $O\left(D\log\left({G}\right)\log\log\left({G}\right)\polylog\left(\frac{1}{\Delta}\right)\right)$ gates, one query to the quantum oracle $\mathcal{O}_\rho^\dag$, i.e., the inverse of $\mathcal{O}_\rho$, and one query to the quantum oracle $\mathcal{O}_\tau$.
\end{lemma}

\begin{figure}[t]
  \centering
  \includegraphics[width=5.5in]{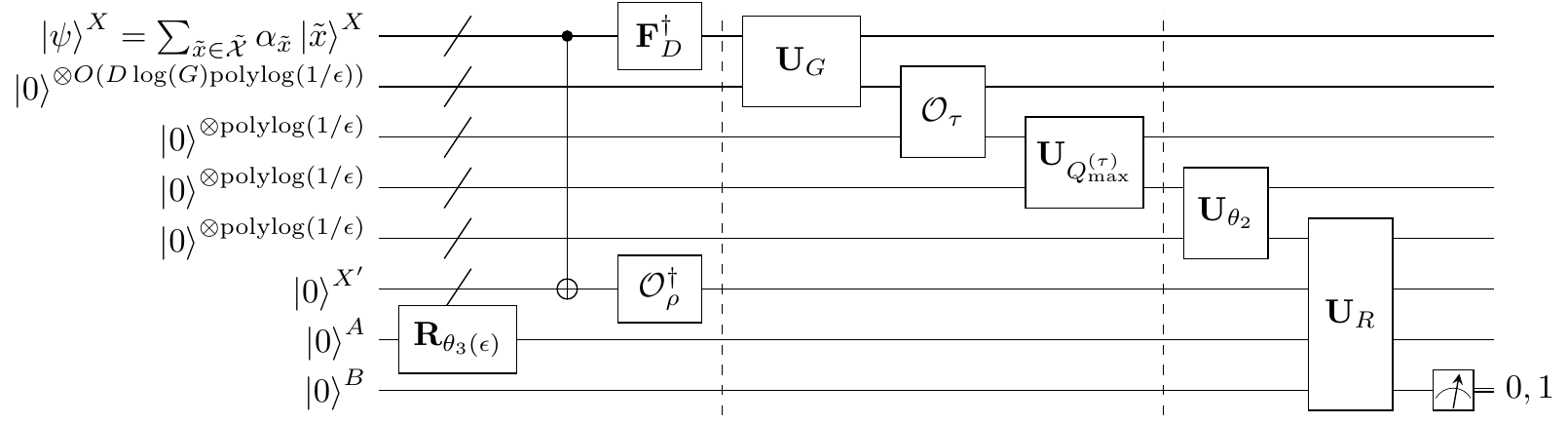}
  \caption{\label{fig:sigma_lambda}A quantum circuit representing a unitary operator $\U$ that achieves~\eqref{seq:povm_implementation} for $\Lam=\hat{\SigmaOp}_\epsilon$, which can be used for implementing a block encoding of $\hat{\SigmaOp}_\epsilon$. The notations are the same as those in Fig.~\ref{fig:q_tau}. The first gate acting on two quantum registers $X$ and $X^\prime$ collectively represents $\textsc{CNOT}$ gates acting transversally on each of the qubits of these registers. A part of this circuit sandwiched by two vertical dashed lines is the same as the corresponding part in Fig.~\ref{fig:q_tau}.
  Additionally, the circuit performs a preprocessing of the input state before performing the part corresponding to Fig.~\ref{fig:q_tau}, which achieves the transformation of quantum states shown in a chain starting from~\eqref{seq:state_transformation}. Also, the circuit performs the final gates $\U_{\theta_2}$ and $\U_R$ after the part corresponding to Fig.~\ref{fig:q_tau}, which are followed by a measurement described by the analysis starting from~\eqref{seq:probability}.}
\end{figure}

\begin{proof}
  We construct a quantum circuit representing a unitary operator $\U$ that achieves~\eqref{seq:povm_implementation} for $\Lam=\hat{\SigmaOp}_\epsilon$.
  We use the same notations as those in the proof of Lemma~\ref{sprp:block_encoding_q_tau} except the following notations.
  The input state is written in this proof as
  \begin{equation}
    \label{seq:psi}
    \Ket{\psi}^X=\sum_{\tilde{x}\in\tilde{\mathcal{X}}}\alpha_{\tilde{x}}\Ket{\tilde{x}}^X\in\mathcal{H}^X.
  \end{equation}
  Define functions
  \begin{align}
    \theta_2\left(q\right)&\coloneqq\arccos\left(\sqrt{q}\right),\\
    \theta_3\left(\epsilon\right)&\coloneqq\arccos\left(\sqrt{\frac{\epsilon/Q_{\max}^{(\tau)}}{1+\left(\epsilon/Q_{\max}^{(\tau)}\right)}}\right).
  \end{align}
  Define a unitary operator $\U_{\theta_2}$ acting as
  \begin{align}
    \label{seq:U_theta_2}
    \U_{\theta_2}:&\Ket{x}\otimes\Ket{0}\xrightarrow{U_{\theta_2}}\Ket{x}\otimes\Ket{\theta_2(x)}.
  \end{align}
  Let $\U_\rho$ denote a unitary operator representing a quantum circuit for implementing the oracle $\mathcal{O}_\rho$.
  Then, the unitary operator representing its inverse $\mathcal{O}_\rho^\dag$ is given by $\U_\rho^\dag$.
  Define a unitary operator
  \begin{align}
    \label{seq:controlled_rotation_sigma_lambda}
    \U_R=&\left(\mathbbm{1}\otimes\mathbbm{1}^{X^\prime}\otimes\Ket{0}\Bra{0}^A\otimes\mathbbm{1}^B\right)\nonumber\\
         &\quad+\left(\mathbbm{1}\otimes\left(\mathbbm{1}^{X^\prime}-\Ket{0}\Bra{0}^{X^\prime}\right)\otimes\Ket{1}\Bra{1}^A\otimes \sigma_x^B\right)\nonumber\\
    &\quad+\left(\sum_{\theta}\Ket{\theta}\Bra{\theta}\otimes\Ket{0}\Bra{0}^{X^\prime}\otimes\Ket{1}\Bra{1}^A\otimes \R_\theta^B\right),
  \end{align}
  where the first quantum register may store a real number $\theta$ in the fixed-point number representation with sufficient precision to achieve the overall precision $\Delta$ in~\eqref{seq:povm_implementation}, the second quantum register $\mathcal{H}^{X^\prime}$ is isomorphic to the quantum register $\mathcal{H}^X$, i.e., is composed of the same number of qubits as $\mathcal{H}^X$, the third quantum register $\mathcal{H}^A$ is one auxiliary qubit, and the fourth quantum register $\mathcal{H}^B$ is another auxiliary qubit.
  The operators  $\sigma_x^B$ and $\R_\theta^B$ on $\mathcal{H}^B$ are defined as~\eqref{seq:sigma_x} and~\eqref{seq:rotation}, respectively.
  If the state of $A$ is $\Ket{0}^A$, the first term of~\eqref{seq:controlled_rotation_sigma_lambda} does not change the state on $B$, and if $\Ket{1}^A$, the second and third terms of~\eqref{seq:controlled_rotation_sigma_lambda} act as follows:
  unless the state of $X^\prime$ is $\Ket{0}^{X^\prime}$,
  $\sigma_x^B$ in the second term of~\eqref{seq:controlled_rotation_sigma_lambda} flips $\Ket{0}^B$ to $\Ket{1}^B$,
  and if the state of $X^\prime$ is $\Ket{0}^{X^\prime}$,
  $\R_\theta^B$ in the third term of~\eqref{seq:controlled_rotation_sigma_lambda} acts in the same way as~\eqref{seq:4}.

  Using these notations, we show a quantum circuit representing $\U$ in Fig.~\ref{fig:sigma_lambda}.
  While a part of the circuit in Fig.~\ref{fig:sigma_lambda} sandwiched by two vertical dashed lines is the same as the corresponding part in Fig.~\ref{fig:q_tau},
  this circuit additionally performs a preprocessing of the input state before performing the part corresponding to Fig.~\ref{fig:q_tau}, and the final gates $\U_{\theta_2}$ and $\U_R$ in Fig.~\ref{fig:sigma_lambda} after the part corresponding to Fig.~\ref{fig:q_tau} are also different.
  This preprocessing implements the following transformation with sufficient precision $O(\Delta)$ to achieve the overall precision $\Delta$
  \begin{align}
    \label{seq:state_transformation}
    &\Ket{\psi}^X\otimes\Ket{0}\otimes\Ket{0}\otimes\Ket{0}\otimes\Ket{0}\otimes\Ket{0}^{X^\prime}\otimes\Ket{0}^A\otimes\Ket{0}^B\\
    \label{seq:cnot_used}
    \xrightarrow{\textsc{CNOT}}&\sum_{\tilde{x}\in\tilde{\mathcal{X}}}\alpha_{\tilde{x}}\Ket{\tilde{x}}^X\otimes\Ket{0}\otimes\Ket{0}\otimes\Ket{0}\otimes\Ket{0}\otimes\Ket{\tilde{x}}^{X^\prime}\otimes\Ket{0}^A\otimes\Ket{0}^B\\
    \label{seq:f_d_used}
    \xrightarrow{\F_{D}^\dag\otimes\mathcal{O}_\rho^\dag\otimes \R_{\theta_3\left(\epsilon\right)}}&\sum_{\tilde{x}\in\tilde{\mathcal{X}}}\alpha_{\tilde{x}}\F_{D}^\dag\Ket{\tilde{x}}^X\otimes\Ket{0}\otimes\Ket{0}\otimes\Ket{0}\otimes\Ket{0}\otimes \U_\rho^\dag\Ket{\tilde{x}}^{X^\prime}\otimes\nonumber\\
                                                                                                             &\quad\left(\sqrt{\frac{\epsilon/Q_{\max}^{(\tau)}}{1+\left(\epsilon/Q_{\max}^{(\tau)}\right)}}\Ket{0}^A+\sqrt{\frac{1}{1+\left(\epsilon/Q_{\max}^{(\tau)}\right)}}\Ket{1}^A\right)\otimes\Ket{0}^B.
  \end{align}
  In~\eqref{seq:cnot_used}, we use $O\left(D\log\left({G}\right)\right)$ $\textsc{CNOT}$ gates acting on each of the $O\left(D\log\left({G}\right)\right)$ qubits of the quantum registers $X$ and $X^\prime$.
  In~\eqref{seq:f_d_used}, $\F_{D}^\dag$ is implemented by $O\left(D\log\left({G}\right)\log\left(\frac{\log G}{\Delta}\right)\right)$ gates as shown in~\eqref{seq:runtime_F_D},
  $\mathcal{O}_\rho^\dag$ takes time $T_\rho$, and a fixed one-qubit rotation $\R_{\theta_3\left(\epsilon\right)}$ defined as~\eqref{seq:rotation} is implemented with precision $O(\Delta)$ using $O\left(\polylog\left(\frac{1}{\Delta}\right)\right)$ gates~\cite{N4}.

  Then, after performing the same part as in Fig.~\ref{fig:q_tau}, which is dominant, the circuit in Fig.~\ref{fig:sigma_lambda} performs $\U_{\theta_2}$ defined as~\eqref{seq:U_theta_2} and $\U_R$ defined as~\eqref{seq:controlled_rotation_sigma_lambda}.
  We can implement $\U_{\theta_2}$ in the same way as~\eqref{seq:3}, i.e., $\U_{\theta_1}$ in Lemma~\ref{sprp:block_encoding_q_tau}, using $O\left(\polylog\left(\frac{1}{\Delta}\right)\right)$ gates.
  We can implement $\U_R$ using $O\left(D\log\left({G}\right)\polylog\left(\frac{1}{\Delta}\right)\right)$ gates since $\Ket{\theta}$ is stored in $O\left(\polylog\left(\frac{1}{\Delta}\right)\right)$ qubits and $\mathcal{H}^{X^\prime}$ consists of $O\left(D\log\left({G}\right)\right)$ qubits.

  After performing $\U_R$, we perform a measurement of the last qubit $\mathcal{H}^B$ in the computational basis $\{\Ket{0}^B,\Ket{1}^B\}$.
  To calculate the probability of obtaining the outcome $0$ in this measurement of $\mathcal{H}^B$,
  suppose that we performed a measurement of the one-qubit register $\mathcal{H}^A$ in the computational basis $\left\{\Ket{0}^A,\Ket{1}^A\right\}$.
  Then, we would obtain the outcome $\Ket{0}^A$ with probability $\frac{\epsilon/Q_{\max}^{(\tau)}}{1+\left(\epsilon/Q_{\max}^{(\tau)}\right)}$, and the outcome $\Ket{1}^A$ with probability $\frac{1}{1+\left(\epsilon/Q_{\max}^{(\tau)}\right)}$.
  Conditioned on the outcome $\Ket{0}^A$, the measurement of $\mathcal{H}^B$ yields the outcome $\Ket{0}^B$ with probability $1$ correspondingly to the first term of~\eqref{seq:controlled_rotation_sigma_lambda}.
  Conditioned on $\Ket{1}^A$, owing to the third term of~\eqref{seq:controlled_rotation_sigma_lambda}, the measurement of $\mathcal{H}^B$ yields the outcome $\Ket{0}^B$ with probability
  \begin{align}
    \label{seq:probability}
    &\sum_{\tilde{x}^\prime\in\tilde{\mathcal{X}}}{\left|\sum_{\tilde{x}\in\tilde{\mathcal{X}}}\alpha_{\tilde{x}}\Braket{\tilde{x}^\prime|\F_{D}^\dag|\tilde{x}}\Braket{0|\U_\rho^\dag|\tilde{x}}\sqrt{\frac{1}{Q_{\max}^{(\tau)}}Q^{(\tau)}\left(\frac{\tilde{x}^\prime}{G}\right)}\right|}^2\nonumber\\
    &=\sum_{\tilde{x}^\prime\in\tilde{\mathcal{X}}}\frac{1}{Q_{\max}^{(\tau)}}Q^{(\tau)}\left(\frac{\tilde{x}^\prime}{G}\right){\left|\sum_{\tilde{x}\in\tilde{\mathcal{X}}}\alpha_{\tilde{x}}\Braket{\tilde{x}^\prime|\F_{D}^\dag|\tilde{x}}\Braket{0|\U_\rho^\dag|\tilde{x}}\right|}^2.
  \end{align}
  Note that the second term of~\eqref{seq:controlled_rotation_sigma_lambda} has no contribution in~\eqref{seq:probability} because $\sigma_x^B$ flips $\Ket{0}^B$ to $\Ket{1}^B$.
  Due to ${\left|\Braket{0|\U_\rho^\dag|\tilde{x}}\right|}={\left|\Bra{\tilde{x}}\left(\sum_{\tilde{x}^\prime}\sqrt{\hat{q}^{(\rho)}\left(\tilde{x}^\prime\right)}\Ket{\tilde{x}^\prime}\right)\right|}=\sqrt{\hat{q}^{(\rho)}(\tilde{x})}$, we have
  \begin{align}
    \label{seq:probability2}\eqref{seq:probability}&=\sum_{\tilde{x}^\prime\in\tilde{\mathcal{X}}}\frac{1}{Q_{\max}^{(\tau)}}Q^{(\tau)}\left(\frac{\tilde{x}^\prime}{G}\right){\left|\sum_{\tilde{x}\in\tilde{\mathcal{X}}}\alpha_{\tilde{x}}\Braket{\tilde{x}^\prime|\F_{D}^\dag|\tilde{x}}\sqrt{\hat{q}^{(\rho)}\left(\tilde{x}\right)}\right|}^2\nonumber\\
    &=\sum_{\tilde{x}^\prime\in\tilde{\mathcal{X}}}\frac{1}{Q_{\max}^{(\tau)}}Q^{(\tau)}\left(\frac{\tilde{x}^\prime}{G}\right){\left|\sum_{\tilde{x}\in\tilde{\mathcal{X}}}\alpha_{\tilde{x}}\Braket{\tilde{x}^\prime|\F_{D}^\dag\sqrt{\hat{\q}^{(\rho)}}|\tilde{x}}\right|}^2.
  \end{align}
  By definition~\eqref{seq:psi} of $\Ket{\psi}$, we have
  \begin{align}
    \label{seq:probability3}\eqref{seq:probability2}&=\sum_{\tilde{x}^\prime\in\tilde{\mathcal{X}}}\frac{1}{Q_{\max}^{(\tau)}}Q^{(\tau)}\left(\frac{\tilde{x}^\prime}{G}\right){\left|\Braket{\tilde{x}^\prime|\F_{D}^\dag\sqrt{\hat{\q}^{(\rho)}}|\psi}\right|}^2\nonumber\\
    &=\sum_{\tilde{x}^\prime\in\tilde{\mathcal{X}}}\frac{1}{Q_{\max}^{(\tau)}}Q^{(\tau)}\left(\frac{\tilde{x}^\prime}{G}\right)\Braket{\tilde{x}^\prime|\F_{D}^\dag\sqrt{\hat{\q}^{(\rho)}}\Ket{\psi}\Bra{\psi}\sqrt{\hat{\q}^{(\rho)}}\F_{D}|\tilde{x}^\prime}\nonumber\\
    &=\frac{1}{Q_{\max}^{(\tau)}}\tr\left[\F_{D}^\dag\sqrt{\hat{\q}^{(\rho)}}\Ket{\psi}\Bra{\psi}\sqrt{\hat{\q}^{(\rho)}}\F_{D}\left(\sum_{\tilde{x}^\prime\in\tilde{\mathcal{X}}}Q^{(\tau)}\left(\frac{\tilde{x}^\prime}{G}\right)\Ket{\tilde{x}^\prime}\Bra{\tilde{x}^\prime}\right)\right].
  \end{align}
  By definition~\eqref{seq:q_tau_g} of $\Q^{(\tau)}$, we obtain
  \begin{align}
    \label{seq:probability4}\eqref{seq:probability3}&=\frac{1}{Q_{\max}^{(\tau)}}\tr\left[\F_{D}^\dag\sqrt{\hat{\q}^{(\rho)}}\Ket{\psi}\Bra{\psi}\sqrt{\hat{\q}^{(\rho)}}\F_{D} \Q^{(\tau)}\right]\nonumber\\
    &=\frac{1}{Q_{\max}^{(\tau)}}\tr\left[\Ket{\psi}\Bra{\psi}\sqrt{\hat{\q}^{(\rho)}}\F_{D} \Q^{(\tau)}\F_{D}^\dag\sqrt{\hat{\q}^{(\rho)}}\right]\nonumber\\
    &=\frac{1}{Q_{\max}^{(\tau)}}\tr\left[\Ket{\psi}\Bra{\psi}\sqrt{\hat{\q}^{\left(\rho\right)}}\kernel\sqrt{\hat{\q}^{\left(\rho\right)}}\right],
  \end{align}
  where the last equality follows from the perfect reconstruction of the kernel $k$ shown in Proposition~\ref{sprp:perfect_reconstruction}.
  Therefore, since a measurement of the auxiliary qubit $\mathcal{H}^A$ in the computational basis $\left\{\Ket{0}^A,\Ket{1}^A\right\}$ yields outcome $0$ and $1$ with probability $\frac{\epsilon/Q_{\max}^{(\tau)}}{1+\left(\epsilon/Q_{\max}^{(\tau)}\right)}$ and $\frac{1}{1+\left(\epsilon/Q_{\max}^{(\tau)}\right)}$ respectively,
  the circuit in Fig.~\ref{fig:sigma_lambda} yields the outcome $0$ with probability
  \begin{align}
    &\frac{\epsilon/Q_{\max}^{(\tau)}}{1+\left(\epsilon/Q_{\max}^{(\tau)}\right)}\times 1+\frac{1}{1+\left(\epsilon/Q_{\max}^{(\tau)}\right)}\times\left(\frac{1}{Q_{\max}^{(\tau)}}\tr\left[\Ket{\psi}\Bra{\psi}\sqrt{\hat{\q}^{\left(\rho\right)}}\kernel\sqrt{\hat{\q}^{\left(\rho\right)}}\right]\right)\nonumber\\
    &=\frac{1}{1+\left(\epsilon/Q_{\max}^{(\tau)}\right)}\times \tr\left[\Ket{\psi}\Bra{\psi}\left(\frac{\epsilon}{Q_{\max}^{(\tau)}}\mathbbm{1}\right)\right]\nonumber\\
    &\quad+\frac{1}{1+\left(\epsilon/Q_{\max}^{(\tau)}\right)}\times\tr\left[\Ket{\psi}\Bra{\psi}\left(\frac{1}{Q_{\max}^{(\tau)}}\sqrt{\hat{\q}^{\left(\rho\right)}}\kernel\sqrt{\hat{\q}^{\left(\rho\right)}}\right)\right]\nonumber\\
    &=\frac{1}{1+\left(\epsilon/Q_{\max}^{(\tau)}\right)}\times \tr\left[\Ket{\psi}\Bra{\psi}\left(\frac{1}{Q_{\max}^{(\tau)}}\sqrt{\hat{\q}^{\left(\rho\right)}}\kernel\sqrt{\hat{\q}^{\left(\rho\right)}}+\frac{\epsilon}{Q_{\max}^{(\tau)}}\mathbbm{1}\right)\right]\nonumber\\
    &=\tr\left[\Ket{\psi}\Bra{\psi}\left(\frac{1}{1+\left(\epsilon/Q_{\max}^{(\tau)}\right)}\hat{\SigmaOp}_\epsilon\right)\right],
  \end{align}
  where the last equality follows from the definition~\eqref{seq:hat_Sigma_lambda} of $\hat{\SigmaOp}_\epsilon$,
  which achieves~\eqref{seq:povm_implementation} for $\Lam=\frac{1}{1+\left(\epsilon/Q_{\max}^{(\tau)}\right)}\hat{\SigmaOp}_\epsilon$ within a claimed runtime.
\end{proof}

Using the block encodings in Lemmas~\ref{sprp:block_encoding_q_tau} and~\ref{sprp:block_encoding_sigma_lambda}, we prove Theorem~\ref{sthm:sampling} as follows.

\begin{proof}[Proof of Theorem~\ref{sthm:sampling}]
  We prove that Algorithm~\ref{salg:random_feature_revised} has the claimed runtime guarantee. The dominant step of Algorithm~\ref{salg:random_feature_revised} is Step~5, as shown in the following.

  In Step~2, after the initialization of $\Ket{0}^X\otimes\Ket{0}^{X^\prime}$, we prepare $\sum_{\tilde{x}\in\tilde{\mathcal{X}}}\Ket{0}^X\otimes\sqrt{\hat{\q}^{(\rho)}}\Ket{\tilde{x}}^{X^\prime}$ by one query to the oracle $\mathcal{O}_\rho$ defined as~\eqref{seq:oracle_rho_definition}, followed by $O\left(D\log\left({G}\right)\right)$ $\textsc{CNOT}$ gates to prepare $\sum_{\tilde{x}\in\tilde{\mathcal{X}}}\Ket{\tilde{x}}^X\otimes\sqrt{\hat{\q}^{(\rho)}}\Ket{\tilde{x}}^{X^\prime}$, since $\mathcal{H}^X$ consists of $O\left(D\log\left({G}\right)\right)$ qubits.
  Step~3 performs $\F_{D}^\dag$, which is implemented using $O\left(D\log\left({G}\right)\log\left(\frac{\log G}{\Delta}\right)\right)$ gates as shown in~\eqref{seq:runtime_F_D}.
  Step~4 is implemented by the block encoding of $\sqrt{\frac{1}{Q_{\max}^{(\tau)}}\Q^{(\tau)}}$ within time $O\left(D\log\left({G}\right)\log\log\left({G}\right)\polylog\left(\frac{1}{\Delta}\right)+T_\tau\right)$ as shown in Lemma~\ref{sprp:block_encoding_q_tau}.
  The runtime at this moment is $O\left(D\log\left({G}\right)\log\log\left({G}\right)\polylog\left(\frac{1}{\Delta}\right)+T_\rho+T_\tau\right)$.
  After applying the block encoding of $\sqrt{\frac{1}{Q_{\max}^{(\tau)}}\Q^{(\tau)}}$, we obtain a quantum state represented as a linear combination including a term
  \begin{equation}
    \label{seq:state_to_amplify}
    \sum_{\tilde{x}\in\tilde{\mathcal{X}}}\Ket{\tilde{x}}^X\otimes \sqrt{\frac{1}{Q_{\max}^{(\tau)}}\Q^{(\tau)}}\F_{D}^\dag\sqrt{\hat{\q}^{\left(\rho\right)}}\Ket{\tilde{x}}^{X^\prime},
  \end{equation}
  and the norm of this term is
  \begin{align}
  &\left\|\sum_{\tilde{x}\in\tilde{\mathcal{X}}}\Ket{\tilde{x}}^X\otimes \sqrt{\frac{1}{Q_{\max}^{(\tau)}}\Q^{(\tau)}}\F_{D}^\dag\sqrt{\hat{\q}^{\left(\rho\right)}}\Ket{\tilde{x}}^{X^\prime}\right\|_2\nonumber\\
  &=\sqrt{\frac{\tr\left[\sqrt{\hat{\q}^{\left(\rho\right)}} \F_{D} \Q^{(\tau)}\F_{D}^\dag\sqrt{\hat{\q}^{\left(\rho\right)}}\right]}{Q_{\max}^{(\tau)}}}\nonumber\\
  &=\sqrt{\frac{\tr\left[\F_{D} \Q^{(\tau)}\F_{D}^\dag\hat{\q}^{\left(\rho\right)}\right]}{Q_{\max}^{(\tau)}}}\nonumber\\
  &=\sqrt{\frac{\tr\hat{\SigmaOp}}{Q_{\max}^{(\tau)}}},
  \end{align}
  where the last equality uses $\hat{\SigmaOp}=\kernel\hat{\q}^{(\rho)}=\F_D \Q^{(\tau)} \F_D^\dag \hat{\q}^{(\rho)}$ obtained from Proposition~\ref{sprp:perfect_reconstruction}.
  For any translation-invariant kernel $\tilde{k}\left(x^\prime,x\right)=\tilde{k}_\mathrm{TI}\left(x^\prime-x\right)$, we can evaluate $\tr\hat{\SigmaOp}$ as
  \begin{equation}
    \tr\hat{\SigmaOp}=\tr \left[\kernel \hat{\q}^{\left(\rho\right)}\right]=\tilde{k}_\mathrm{TI}\left(0\right)\tr \hat{\q}^{\left(\rho\right)}=\tilde{k}(0,0)=\Omega(1),
  \end{equation}
  where we use the assumption $\tilde{k}(0,0)=\Omega(k(0,0))=\Omega(1)$.
  Thus, to obtain the normalized quantum state proportional to the term~\eqref{seq:state_to_amplify}, Step~4 is followed by amplitude amplification~\cite{BrassardHoyer} that repeats the above steps $O\left(\sqrt{\frac{Q_{\max}^{(\tau)}}{\tr\hat{\SigmaOp}}}\right)=O\left(\sqrt{Q_{\max}^{(\tau)}}\right)$ times.
  Therefore, at the end of Step~4 including the amplitude amplification, the runtime is
  \begin{equation}
    \label{seq:step_4}
    O\left(\left(D\log\left({G}\right)\log\log\left({G}\right)\polylog\left(\frac{1}{\Delta}\right)+T_\rho+T_\tau\right)\times\sqrt{Q_{\max}^{(\tau)}}\right).
  \end{equation}

  Step~5 is performed by implementing a block encoding of $\hat{\SigmaOp}_\epsilon^{-\frac{1}{2}}$, which is obtained from quantum singular value transformation (QSVT)~\cite{G1} of the block encoding of $\frac{1}{1+(\epsilon/Q_{\max}^{(\tau)})}\hat{\SigmaOp}_\epsilon$ constructed in Lemma~\ref{sprp:block_encoding_sigma_lambda}.
  The block encoding of $\frac{1}{1+(\epsilon/Q_{\max}^{(\tau)})}\hat{\SigmaOp}_\epsilon$ can be implemented in time
  $O\left(D\log\left({G}\right)\log\log\left({G}\right)\polylog\left(\frac{1}{\Delta}\right)+T_\rho+T_\tau\right)$
  as shown in Lemma~\ref{sprp:block_encoding_sigma_lambda}.
  Then, the QSVT combined with variable-time amplitude amplification~\cite{ambainis2012variable,childs2017quantum,chakraborty2018} yields a block encoding of ${\left(\frac{1}{1+(\epsilon/Q_{\max}^{(\tau)})}\hat{\SigmaOp}_\epsilon\right)}^{-\frac{1}{2}}$, which can be applied to any given quantum state up to $\Delta$ precision using the block encoding of $\frac{1}{1+(\epsilon/Q_{\max}^{(\tau)})}\hat{\SigmaOp}_\epsilon$ repeatedly $\widetilde{O}\left({\left(\frac{Q_{\max}^{(\tau)}}{\epsilon}+1\right)}\polylog\left(\frac{1}{\Delta}\right)\right)$ times~\cite{G1}.
  This repetition includes the runtime required for the amplitude amplification,
  and $\frac{Q_{\max}^{(\tau)}}{\epsilon}+1$ is the condition number of $\frac{1}{1+(\epsilon/Q_{\max}^{(\tau)})}\hat{\SigmaOp}_\epsilon$ since it holds that
  \begin{equation}
    \frac{1}{1+(\epsilon/Q_{\max}^{(\tau)})}\frac{\epsilon}{Q_{\max}^{(\tau)}}\mathbbm{1}\leqq\frac{1}{1+(\epsilon/Q_{\max}^{(\tau)})}\hat{\SigmaOp}_\epsilon\leqq\mathbbm{1}.
  \end{equation}
  Thus, Step~5 including amplitude amplification can be implemented in time
  \begin{align}
    \label{seq:step_5}
  &O\left(D\log\left({G}\right)\log\log\left({G}\right)\polylog\left(\frac{1}{\Delta}\right)+T_\rho+T_\tau\right)\times\widetilde{O}\left({\left(\frac{Q_{\max}^{(\tau)}}{\epsilon}+1\right)}\polylog\left(\frac{1}{\Delta}\right)\right)\nonumber\\
  &=O\left(D\log\left({G}\right)\log\log\left({G}\right)+T_\rho+T_\tau\right)\times\widetilde{O}\left(\frac{Q_{\max}^{(\tau)}}{\epsilon}\polylog\left(\frac{1}{\Delta}\right)\right).
  \end{align}
  Therefore from~\eqref{seq:step_4} and~\eqref{seq:step_5}, we obtain the total runtime at the end of Step~5 including amplitude amplification
  \begin{align}
  &O\left(D\log\left({G}\right)\log\log\left({G}\right)+T_\rho+T_\tau\right)\times\widetilde{O}\left({\left(\sqrt{Q_{\max}^{(\tau)}}+\frac{Q_{\max}^{(\tau)}}{\epsilon}\right)}\polylog\left(\frac{1}{\Delta}\right)\right)\nonumber\\
  &=O\left(D\log\left({G}\right)\log\log\left({G}\right)+T_\rho+T_\tau\right)\times\widetilde{O}\left(\frac{Q_{\max}^{(\tau)}}{\epsilon}\polylog\left(\frac{1}{\Delta}\right)\right),
  \end{align}
  which yields the conclusion.
\end{proof}

\section{\label{sec:4}Overall runtime of learning with optimized random features}

In this section, we show the algorithm for the learning with optimized random features and bound its runtime.

Algorithm~\ref{salg:data_approximation} shows the whole algorithm that achieves the learning using the optimized random features.
In Algorithm~\ref{salg:data_approximation},
we sample the optimized random features efficiently by Algorithm~\ref{salg:random_feature_revised}, and then perform stochastic gradient descent (SGD)~\cite{pmlr-v99-jain19a} shown in Algorithm~\ref{salg:sgd}.
As explained in the main text,
the SGD achieves linear regression to obtain coefficients of the estimate $\hat{f}_{M,v_m,\alpha_m}=\sum_{m=0}^{M-1}\alpha_m\varphi(v_m,\cdot)\approx f$, i.e.,
\begin{equation}
  \alpha=\left(\begin{matrix}\alpha_0\\
      \vdots\\
  \alpha_{M-1}\end{matrix}\right)\in\mathbb{R}^M,
\end{equation}
where the optimal coefficient $\alpha$ minimizes the generalization error
\begin{equation}
  \label{seq:generalization_error}
  I\left(\alpha\right)\coloneqq\sum_{\tilde{x}\in\tilde{\mathcal{X}}}p^{(\rho)}(\tilde{x})\Big|f(\tilde{x})-\sum_{m=0}^{M-1}\alpha_{m}\varphi\left(v_m,\tilde{x}\right)\Big|^2,
\end{equation}
and the examples of data are IID sampled according to $p^{(\rho)}(\tilde{x})\coloneqq \int_{\Delta_{\tilde{x}}}d\rho(x)$.

We remark that rather than the linear regression based on least-squares of $I$ in~\eqref{seq:generalization_error},~\citet{B1} analyzes regularized least-squares regression exploiting $Q_\epsilon^\ast$, but it may be hard to compute description of $Q_\epsilon^\ast$. To circumvent this hardness of using $Q_\epsilon^\ast$ in regularization, we could replace the regularization in Ref.~\cite{B1} with $L_2$ regularization $R(\alpha)=\lambda\|\alpha\|_2^2$. Then, due to strong convexity~\cite{pmlr-v99-jain19a}, SGD reducing $I+R$ to $O(\epsilon)$ terminates after $O(\frac{1}{\epsilon\lambda})$ iterations, while further research is needed to clarify how the $L_2$ regularization affects the learning accuracy compared to minimizing $I$ without this regularization.
In this paper, we consider the linear regression minimizing $I$ to simplify the analysis of the required runtime for achieving the desired learning accuracy.

We prove the following theorem that bounds the runtime of Algorithm~\ref{salg:data_approximation}.
The sketch of the proof is as follows.
To bound the runtime of Algorithm~\ref{salg:data_approximation},
we show that the required number $T$ of iterations for the SGD~\cite{pmlr-v99-jain19a}, i.e., Algorithm~\ref{salg:sgd}, to return $\alpha$ minimizing $I$ to accuracy $O(\epsilon)$ with high probability greater than $1-\delta$ is
\begin{equation}
\label{seq:T}
T=O\left(\frac{1}{\epsilon^2 Q_{\min}^2}\log\left(\frac{1}{\delta}\right)\right),
\end{equation}
where $Q_{\min}$ is the minimum of $Q(v_0),\ldots,Q(v_{M-1})$ in Theorem~\ref{sthm:sampling}
\begin{equation}
\label{seq:Q_min}
  Q_{\min}\coloneqq\min\left\{Q(v_m):m\in\left\{0,\ldots,M-1\right\}\right\},
\end{equation}
and the parameter region $\mathcal{W}$ of $\alpha$ in Algorithm~\ref{salg:sgd} is chosen as an $M$-dimensional ball of center $0$ and of radius $O\left(\frac{1}{\sqrt{MQ_{\min}}}\right)$.
Note that step sizes used for SGD~\cite{pmlr-v99-jain19a} shown in Algorithm~\ref{salg:sgd} are chosen depending on the number $T$ of iterations so as to achieve~\eqref{seq:T}, but we can also use step sizes independent of $T$ at expense of as small as poly-logarithmic slowdown in terms of $\epsilon$, $Q_{\min}$ compared to~\eqref{seq:T}~\cite{H3}.
In the $t$th iteration of the SGD for each $t\in\left\{1,\ldots,T\right\}$,
we calculate an unbiased estimate $\hat{g}^{\left(t\right)}$ of the gradient $\nabla I$.
Using the $t$th IID sampled data $\left(\tilde{x}_t,y_t\right)$ and a uniformly sampled random integer $m\in\left\{0,\ldots,M-1\right\}$,
we show that we can calculate $\hat{g}^{\left(t\right)}$ within time $O(MD)$ in addition to one query to each of the classical oracles
\begin{equation}
  \label{seq:classical_oracles}
 \mathcal{O}_{\tilde{x}}(n)=\tilde{x}_n,\quad \mathcal{O}_y(n)=y_n
\end{equation}
to get $\left(\tilde{x}_t,y_t\right)$ in time $T_{\tilde{x}}$ and $T_y$, respectively; that is, the runtime per iteration of the SGD is $O(MD+T_{\tilde{x}}+T_y)$.
Combining Algorithm~\ref{salg:random_feature_revised} with this SGD, we achieve the learning by Algorithm~\ref{salg:data_approximation} within the following overall runtime.

\begin{theorem}
  [\label{sthm:runtime_all}Overall runtime of learning with optimized random features]
  The runtime $T_2$ of Algorithm~\ref{salg:data_approximation} for learning with optimized random features is
  \begin{equation*}
    T_2=O\left(MT_1\right)+O\left(\left(MD+T_{\tilde{x}}+T_y\right)\frac{1}{{\epsilon}^2 Q_{\min}^2}\log\left(\frac{1}{\delta}\right)\right),
  \end{equation*}
  where $T_1$ appears in Theorem~\ref{sthm:sampling}, the first term is the runtime of sampling $M$ optimized random features by Algorithm~\ref{salg:random_feature_revised}, and the second term is runtime of the SGD\@.
\end{theorem}

\begin{remark}
[Omission of some parameters from the informal statement of Theorem~\ref{sthm:runtime_all} in the main text]
  In the informal statement of Theorem~\ref{sthm:runtime_all} in the main text, we omit the dependency on $Q_{\min}$ and $\frac{1}{\delta}$ from the runtime.
  In the parameter region of sampling optimized random features that are weighted by importance and that nearly minimize the required number $M$ of features, the minimal weight $Q_{\min}$ of these features is expected to be sufficiently large compared to $\epsilon$, not dominating the runtime, while we include $Q_{\min}$ in our runtime analysis in Supplementary Material to bound the worst-case runtime.
  In addition, the dependency on $\frac{1}{\delta}$ is logarithmic in Theorem~\ref{sthm:runtime_all}.
  For these reasons, we simplify the presentation in the main text by  omitting $Q_{\min}$ and $\frac{1}{\delta}$.
\end{remark}

\begin{algorithm}[H]
  \caption{\label{salg:data_approximation}Algorithm for learning with optimized random features.}
  \begin{algorithmic}[1]
    \REQUIRE{Inputs to Algorithms~\ref{salg:random_feature_revised} and~\ref{salg:sgd}, required number $M$ of features for achieving  the learning to accuracy $O({\epsilon})$.}
    \ENSURE{Optimized random features $v_0,\ldots,v_{M-1}$ and coefficients $\alpha_0,\ldots,\alpha_{M-1}$ for $\sum_m\alpha_m\varphi(v_m,\cdot)$ to achieve the learning to accuracy $O(\epsilon)$ with probability greater than $1-\delta$.}
    \FOR{$m\in\left\{0,\ldots,M-1\right\}$}
    \STATE{$v_m\gets\textrm{quOptRF}$.}~\COMMENT{by Algorithm~\ref{salg:random_feature_revised}.}
    \ENDFOR%
    \STATE{Minimize $I(\alpha)$ to accuracy $O({\epsilon})$ by \textrm{SGD} to obtain $\alpha_0,\ldots,\alpha_{M-1}$.}~\COMMENT{by Algorithm~\ref{salg:sgd}.}
    \STATE{\textbf{Return }$v_0,\ldots,v_{M-1},\alpha_0,\ldots,\alpha_{M-1}$.}
  \end{algorithmic}
\end{algorithm}

\begin{algorithm}[H]
  \caption{\label{salg:sgd}Stochastic gradient descent (SGD).}
  \begin{algorithmic}[1]
    \REQUIRE{A function $I:\mathcal{W}\to\mathbb{R}$, a projection $\Pi$ to a convex parameter region $\mathcal{W}\subset\mathbb{R}^M$ specified by $Q_{\min}$ in~\eqref{seq:Q_min}, number of iterations $T\in\mathbb{N}$ specified by~\eqref{seq:T}, an initial point $\alpha^{\left(1\right)}\in\mathcal{W}$, $T$-dependent hyperparameters representing step sizes $\left(\eta^{\left(t\right)}:t=1,\ldots,T\right)$ given in Ref.~\cite{pmlr-v99-jain19a}, classical oracle functions $\mathcal{O}_{\tilde{x}},\mathcal{O}_y$ in~\eqref{seq:classical_oracles} for calculating $\hat{g}^{(t)}$.}
    \ENSURE{Approximate solution $\alpha$ minimizing $I(\alpha)$.}
    \FOR{$t\in\left\{1,\ldots,T\right\}$}
    \STATE{Calculate an unbiased estimate $\hat{g}^{\left(t\right)}$ of the gradient of $I$ satisfying $\mathbb{E}\left[\hat{g}^{(t)}\right]=\nabla I(\alpha^{(t)})$.}
    \STATE{$\alpha^{\left(t+1\right)}\gets\Pi(\alpha^{\left(t\right)}-\eta^{\left(t\right)} \hat{g}^{\left(t\right)})$.}
    \ENDFOR%
    \STATE{\textbf{Return } $\alpha\gets\alpha^{(T+1)}$.}
  \end{algorithmic}
\end{algorithm}

\begin{proof}
  We bound the runtime of each step of Algorithm~\ref{salg:data_approximation}.
  In Step~2, using Algorithm~\ref{salg:random_feature_revised} repeatedly $M$ times, we obtain $M$ optimized random features within time
  \begin{equation}
    \label{seq:algorithm_1}
    O(MT_1),
  \end{equation}
  where $T_1$ is the runtime of Algorithm~\ref{salg:random_feature_revised} given by Theorem~\ref{sthm:sampling}.
  As for Step~4, we bound the runtime of the SGD in Algorithm~\ref{salg:sgd}.
  In the following, we show that the runtime of each iteration of the SGD is $O\left(MD+T_{\tilde{x}}+T_y\right)$, and the required number of iterations in the SGD is upper bounded by $O\left(\frac{1}{{\epsilon}^2 Q_{\min}^2}\log\left(\frac{1}{\delta}\right)\right)$.

  We analyze the runtime of each iteration of the SGD\@.
  The dominant step in the $t$th iteration for each $t\in\left\{0,\ldots,T-1\right\}$ is the calculation of an unbiased estimate $\hat{g}^{(t)}$ of the gradient $\nabla I$, where $I$ is given by~\eqref{seq:generalization_error}.
The gradient of $I$ is given by
\begin{align}
  \label{seq:nabla_I}
  \nabla I\left(\alpha\right)&=
  \sum_{\tilde{x}\in\tilde{\mathcal{X}}}p^{(\rho)}\left(\tilde{x}\right)\left(\begin{matrix}2\Re\left[ \mathrm{e}^{-2\pi\mathrm{i}v_0\cdot \tilde{x}}\left(f\left(\tilde{x}\right)-\sum_{m=0}^{M-1}\alpha_m\mathrm{e}^{2\pi\mathrm{i}v_m\cdot \tilde{x}}\right)\right]\\
                 \vdots\\
                 2\Re\left[ \mathrm{e}^{-2\pi\mathrm{i}v_{M-1}\cdot \tilde{x}}\left(f\left(\tilde{x}\right)-\sum_{m=0}^{M-1}\alpha_m\mathrm{e}^{2\pi\mathrm{i}v_m\cdot \tilde{x}}\right)\right]
         \end{matrix}\right)\nonumber\\
                             &=\sum_{m=0}^{M-1}\frac{1}{M}\sum_{\tilde{x}\in\tilde{\mathcal{X}}}p^{(\rho)}\left(\tilde{x}\right)\left(\begin{matrix}2\Re\left[ \mathrm{e}^{-2\pi\mathrm{i}v_0\cdot \tilde{x}}\left(f\left(\tilde{x}\right)-M\alpha_m\mathrm{e}^{2\pi\mathrm{i}v_m\cdot \tilde{x}}\right)\right]\\
                 \vdots\\
                 2\Re\left[\mathrm{e}^{-2\pi\mathrm{i}v_{M-1}\cdot \tilde{x}}\left(f\left(\tilde{x}\right)-M\alpha_m\mathrm{e}^{2\pi\mathrm{i}v_m\cdot \tilde{x}}\right)\right]
         \end{matrix}\right),
\end{align}
where $\Re$ represents the real part.
In the $t$th iteration, Algorithm~\ref{salg:sgd} estimates the gradient at a point denoted by
\begin{equation}
  \alpha^{(t)}=\left(\begin{matrix}
      \alpha_{0}^{(t)}\\
      \vdots\\
      \alpha_{M-1}^{(t)}
  \end{matrix}\right)\in\left\{\alpha^{\left(1\right)},\ldots,\alpha^{\left(T\right)}\right\}.
\end{equation}

Using a pair of given data points $\left(\tilde{x}_t,y_t=f\left(\tilde{x}_t\right)\right)\in\left\{\left(\tilde{x}_{0},y_{0}\right),\left(\tilde{x}_{1},y_{1}\right),\ldots,\right\}$ sampled with probability $p^{(\rho)}\left(\tilde{x}\right)$ as observations of an independently and identically distributed (IID) random variable, and an integer $m\in\left\{0,\ldots,M-1\right\}$ uniformly sampled with probability $\frac{1}{M}$, we give an unbiased estimate $\hat{g}^{\left(t\right)}$ of this gradient at each point $\alpha^{(t)}$ by
\begin{align}
\label{seq:hat_g_F}
\hat{g}^{\left(t\right)}&=
\left(\begin{matrix}2\Re\left[ \mathrm{e}^{-2\pi\mathrm{i}v_0\cdot \tilde{x}_t}\left(y_t-M\alpha^{(t)}_m\mathrm{e}^{2\pi\mathrm{i}v_m\cdot \tilde{x}_t}\right)\right]\\
                 \vdots\\
                 2\Re\left[\mathrm{e}^{-2\pi\mathrm{i}v_{M-1}\cdot \tilde{x}_t}\left(y_t-M\alpha^{(t)}_m\mathrm{e}^{2\pi\mathrm{i}v_m\cdot \tilde{x}_t}\right)\right]
         \end{matrix}\right).
\end{align}
By construction, we have
\begin{equation}
  \mathbb{E}\left[\hat{g}^{(t)}\right]=\nabla I\left(\alpha^{(t)}\right).
\end{equation}
We obtain $\tilde{x}_t$ using the classical oracle $\mathcal{O}_{\tilde{x}}$ within time $T_{\tilde{x}}$, and $y_t=f\left(\tilde{x}_t\right)$ using the classical oracle $\mathcal{O}_{y}$ within time $T_{y}$.
As for $m$, since we can represent the integer $m$ using $\lceil\log_2\left(M\right)\rceil$ bits,
where $\lceil x\rceil$ is the least integer greater than or equal to $x$,
we can sample $m$ from a uniform distribution using a numerical library for generating a random number within time
$O\left(\polylog\left(M\right)\right)$.
Note that even in case it is expensive to use randomness in classical computation, quantum computation can efficiently sample $m$ of $\lceil\log_2\left(M\right)\rceil$ bits from the uniform distribution within time $O(\log\left(M\right))$.
In this quantum algorithm, $\lceil\log_2\left(M\right)\rceil$ qubits are initially prepared in $\Ket{0}^{\otimes \lceil\log_2\left(M\right)\rceil}$, and the Hadamard gate $H=\frac{1}{\sqrt{2}}\left(\begin{smallmatrix} 1 & 1 \\ 1 & -1 \end{smallmatrix}\right)$ is applied to each qubit to obtain
\begin{equation}
  \frac{1}{\sqrt{2^{\lceil\log_2\left(M\right)\rceil}}}{\left(\Ket{0}+\Ket{1}\right)}^{\otimes \lceil\log_2\left(M\right)\rceil},
\end{equation}
followed by a measurement of this state in the computational basis to obtain a $\lceil\log_2\left(M\right)\rceil$-bit outcome sampled from the uniform distribution.
Given $\tilde{x}_t$, $y_t$, and $m$, we can calculate each of the $M$ element of $\hat{g}$ in~\eqref{seq:hat_g_F} within time $O(D)$ for calculating the inner product of $D$-dimensional vectors, and hence the calculation of all the $M$ elements takes time $O(MD)$.
Therefore, each iteration takes time
\begin{equation}
  \label{seq:sgd_each_iteration}
  O\left(T_{\tilde{x}}+T_y+\polylog(M)+MD\right)=O\left(MD+T_{\tilde{x}}+T_y\right).
\end{equation}
Note that without sampling $m$, we would need $O(M^2 D)$ runtime per iteration because each of the $M$ elements of the gradient in~\eqref{seq:nabla_I} includes the sum over $M$ terms; thus, the sampling of $m$ is crucial for achieving our $O(MD)$ runtime.

To bound the required number of iterations, we use an upper bound of the number of iterations in Algorithm~\ref{salg:sgd} given in Ref.~\cite{pmlr-v99-jain19a}, which shows that if we have the following:
\begin{itemize}
  \item for any $\alpha\in\mathcal{W}$,
\begin{align}
  \left\|\nabla I(\alpha)\right\|_2\leqq L,
\end{align}
\item the unbiased estimate $\hat{g}$ for any point $\alpha\in\mathcal{W}$ almost surely satisfies
\begin{align}
  \left\|\hat{g}\right\|_2\leqq L,
\end{align}
\item the diameter of $\mathcal{W}$ is bounded by
\begin{equation}
  \diam\mathcal{W}\leqq d,
\end{equation}
\end{itemize}
then, after $T$ iterations, with high probability greater than $1-\delta$, Algorithm~\ref{salg:sgd} returns $\alpha$ satisfying
\begin{equation}
  \label{seq:bound_iteration}
  \epsilon= O\left(dL\sqrt{\frac{\log\left(\frac{1}{\delta}\right)}{T}}\right),
\end{equation}
where we write
\begin{equation}
  \epsilon=I(\alpha)-\min_{\alpha\in\mathcal{W}}\left\{I(\alpha)\right\}.
\end{equation}
In the following, we bound $d$ and $L$ in~\eqref{seq:bound_iteration} to clarify the upper bound of the required number of iterations $T$ in our setting.

To show a bound of $d$,
recall the assumption that we are given a sufficiently large number $M$ of features for achieving the learning in our setting.
Then,~\citet{B1} has shown that
with the $M$ features sampled from the weighted probability distribution $Q(v_m)P^{(\tau)}(v_m)$ by Algorithm~\ref{salg:random_feature_revised},
the learning to the accuracy $O(\epsilon)$ can be achieved with coefficients satisfying
\begin{equation}
  \|\beta\|_2^2= O\left(\frac{1}{M}\right),
\end{equation}
where $\beta={(\beta_0,\ldots,\beta_{M-1})}^\mathrm{T}$ is given for each $m$ by
\begin{equation}
  \beta_m=\sqrt{Q(v_m)}\alpha_m.
\end{equation}
This bound yields
\begin{equation}
  \label{seq:bach_bound}
  \sum_{m=0}^{M-1}Q(v_m)\alpha_m^2= O\left(\frac{1}{M}\right).
\end{equation}
In the worst case, a lower bound of the left-hand side is
\begin{equation}
  \label{seq:lower_bound}
  \sum_{m=0}^{M-1}Q(v_m)\alpha_m^2\geqq Q_{\min}\left\|\alpha\right\|_2^2.
\end{equation}
From~\eqref{seq:bach_bound} and~\eqref{seq:lower_bound}, we obtain an upper bound of the norm of $\alpha$ minimizing $I$
\begin{equation}
  \label{seq:alpha_bound}
  \left\|\alpha\right\|_2^2=O\left(\frac{1}{MQ_{\min}}\right).
\end{equation}
Thus, it suffices to choose the parameter region $\mathcal{W}$ of $\alpha$ as an $M$-dimensional ball of center $0$ and of radius $O\left(\frac{1}{\sqrt{MQ_{\min}}}\right)$, which yields the diameter
\begin{equation}
  \label{seq:bound_d}
  d=O\left(\frac{1}{\sqrt{MQ_{\min}}}\right).
\end{equation}

As for a bound of $L$, we obtain from~\eqref{seq:hat_g_F} and~\eqref{seq:alpha_bound}
\begin{equation}
  \left\|\hat{g}\right\|_2=O\left(M\left\|\alpha\right\|_2+\sqrt{M}\right)=O\left(\sqrt{\frac{M}{Q_{\min}}}+\sqrt{M}\right)=O\left(\sqrt{\frac{M}{Q_{\min}}}\right),
\end{equation}
where we take the worst case of small $Q_{\min}$, and we use bounds
\begin{align}
  \sqrt{\sum_{m=0}^{M-1} {\left|M\alpha_m \mathrm{e}^{2\pi\mathrm{i}v_m\cdot \tilde{x}_t}\right|}^2}&=O(M\|\alpha\|_2),\\
  \sqrt{\sum_{m=0}^{M-1} y_t^2}&=O(\sqrt{M}).
\end{align}
Since the upper bound of $\left\|\hat{g}\right\|_2$ is larger than $\left\|\nabla I(\alpha)\right\|_2$,
we have
\begin{equation}
  \label{seq:bound_L}
  L=O\left(\sqrt{\frac{M}{Q_{\min}}}\right).
\end{equation}

Therefore, using~\eqref{seq:bound_d} and~\eqref{seq:bound_L}, we bound the right-hand side of~\eqref{seq:bound_iteration}
\begin{equation}
  \epsilon=O\left(dL\sqrt{\frac{\log\left(\frac{1}{\delta}\right)}{T}}\right)=O\left(\frac{1}{Q_{\min}}\sqrt{\frac{\log\left(\frac{1}{\delta}\right)}{T}}\right).
\end{equation}
Thus, it follows that
\begin{equation}
  \label{seq:sgd_number_of_iterations}
  T=O\left(\frac{1}{\epsilon^2 Q_{\min}^2}\log\left(\frac{1}{\delta}\right)\right).
\end{equation}

Combining~\eqref{seq:algorithm_1},~\eqref{seq:sgd_each_iteration}, and~\eqref{seq:sgd_number_of_iterations}, we obtain the claimed overall runtime.
\end{proof}

\bibliographystyle{unsrtnat}

\bibliography{a}

\end{document}